\pdfoutput=1

\documentclass[a4paper,UKenglish,autoref,cleveref]{lipics-v2021}

\pdfoutput=1 
\hideLIPIcs  

\title{
Intuitionistic G\"odel-L\"ob logic, \`a la Simpson: labelled systems and birelational semantics
} 

\titlerunning{Intuitionistic G\"odel-L\"ob logic, \`a la Simpson} 

\author{Anupam Das}{University of Birmingham, United Kingdom
}{a.das@bham.ac.uk}{}{}

\author{Iris van der Giessen}{University of Birmingham, United Kingdom}{i.vandergiessen@bham.ac.uk}{}{}

\author{Sonia Marin}{University of Birmingham, United Kingdom}{s.marin@bham.ac.uk}{}{}

\authorrunning{A.~Das, I.~van der Giessen, and S.~Marin} 

\Copyright{Anupam Das, Iris van der Giessen, and Sonia Marin} 



\category{} 

\relatedversion{} 

\funding{This work was partially supported by a UKRI Future Leaders Fellowship, `Structure vs Invariants in Proofs', project reference MR/S035540/1.}

\acknowledgements{We would like to thank Marianna Girlando and Jan Rooduijn for several valuable discussions on the topic. We also thank Marianna for her significant input during our reading group sessions on intuitionistic modal logic that led to preliminary ideas that have resulted in this paper.}

\nolinenumbers 

\EventEditors{John Q. Open and Joan R. Access}
\EventNoEds{2}
\EventLongTitle{42nd Conference on Very Important Topics (CVIT 2016)}
\EventShortTitle{CVIT 2016}
\EventAcronym{CVIT}
\EventYear{2016}
\EventDate{December 24--27, 2016}
\EventLocation{Little Whinging, United Kingdom}
\EventLogo{}
\SeriesVolume{42}
\ArticleNo{23}
\usepackage{latexsym}
\usepackage{xcolor}

\usepackage{mathrsfs}
\usepackage{graphicx}
\usepackage{colonequals}
\usepackage{tikz}
\usetikzlibrary{positioning}
\usetikzlibrary{arrows}

\usepackage[inline]{enumitem}

\usepackage{virginialake}
\vlnosmallleftlabels
\newcommand{\vlderivationauxnc}[1]{#1{\box\derboxone}\vlderivationterm}
\newcommand{\vlderivationnc}{\vlderivationinit\vlderivationauxnc}

\newcommand{\demph}[1]{\textbf{#1}}

\newcommand{\red}[1]{{\color{red} #1}}

\newcommand{\impl}{\to}

\newcommand{\calP}{\mathcal{P}}

\newcommand{\calB}{\mathcal{B}}
\newcommand{\Bmod}{\mathcal{B}}
\newcommand{\Mmod}{\mathcal{M}}
\newcommand{\Kmod}{\mathcal{K}}

\newcommand{\dom}{\textnormal{Dom}}

\renewcommand{\emptyset}{\varnothing}

\renewcommand{\models}{\vDash}
\newcommand{\notmodels}{\nvDash}

\newcommand{\sttr}[2]{#1:#2}

\newcommand{\Prop}{{\sf Pr}}
\newcommand{\Var}{{\sf Var}}
\newcommand{\CL}[1]{#1^+}
\newcommand{\Lab}{\Var}

\newcommand{\STsub}[2]{\sttr #1 #2}

\newcommand{\fv}[1]{\Var(#1)}

\newcommand{\relcxt}[1]{\relats R_{#1}}

\newcommand{\proves}{\vdash}
\newcommand{\nwfproves}{\, \text{$\fontdimen14\textfont2=5.5pt \proves^{\hspace{-.3em} \infty}$} }

\newcommand{\K}{\mathsf K}
\newcommand{\Kfour}{\mathsf{K4}}

\newcommand{\GL}{{\sf GL}}

\newcommand{\IK}{{\sf IK}}

\newcommand{\IKfour}{{\sf IK4}}
\newcommand{\IGL}{{\sf IGL}}
\newcommand{\iGL}{{\sf iGL}}

\newcommand{\limp}{\to}

\newcommand{\relats}[1]{\mathbf{#1}}
\newcommand{\seqar}{\Rightarrow}

\newcommand{\trans}{\mathsf{tr}}

\newcommand{\id}{\mathsf{id}}
\newcommand{\cut}{\mathsf{cut}}

\newcommand{\leftrule}[1]{ #1\text{-}l}
\newcommand{\rightrule}[1]{#1\text{-}r}

\newcommand{\wk}{\mathsf{w}}
\newcommand{\contr}{\mathsf c}

\newcommand{\thinning}{\mathsf{th}}

\newcommand{\kaxiom}{({\sf k})}

\newcommand{\fouraxiom}{({\sf 4})}

\newcommand{\wlobaxiom}{({\sf l\ddot{o}b})}

\newcommand{\system}{\mathsf S}
\newcommand{\nwfprfclass}{\mathbf P}

\newcommand{\labelled}{\ell}
\newcommand{\multi}{\mathsf m}

\newcommand{\labK}{\labelled \K}
\newcommand{\labKfour}{\labelled\Kfour}
\newcommand{\labGL}{\labelled\GL}

\newcommand{\intlabK}{\labelled \IK }
\newcommand{\intlabKfour}{\labelled \IKfour}
\newcommand{\intlabGL}{\labelled \IGL}

    \newcommand{\labIK}{\intlabK}
    \newcommand{\labIGL}{\intlabGL}
    \newcommand{\labIKfour}{\intlabKfour}

\newcommand{\multiintlabKfour}{\multi\intlabKfour}
\newcommand{\multiintlabGL}{\multi\intlabGL}

    \newcommand{\mlIGL}{\multiintlabGL}
    \newcommand{\mlIKfour}{\multiintlabKfour}

\newcommand{\LK}{{\sf LK}}

\newcommand{\LJ}{{\sf LJ}}

\newcommand{\disjunctive}{\lor}
\newcommand{\dlabIKfour}{\disjunctive \labIKfour}
\newcommand{\dlabIGL}{\disjunctive\labIGL}

\renewcommand{\phi}{\varphi}

\newcommand{\calI}{\mathcal{I}}
\newcommand{\calJ}{\mathcal{J}}

\newcommand{\birel}{\mathscr B}
\newcommand{\pred}{\mathscr P}

\newcommand{\birelIGL}{\birel\IGL}
\newcommand{\predIGL}{\pred\IGL}

\newcommand{\Fframe}{\mathcal F}
\newcommand{\M}{\mathcal M}
\newcommand{\interp}[2]{#2^{#1}}

\newcommand{\prover}{\mathrm{P}}
\newcommand{\denier}{\mathrm{D}}

\newcommand{\cutred}{\mathsf{r}}
\newcommand{\downset}[1]{\lfloor #1 \rfloor}

\newcommand{\PA}{\mathrm{PA}}
\newcommand{\HA}{\mathrm{HA}}

\begin{document}

\maketitle

\begin{abstract}
We derive an intuitionistic version of Gödel-Löb modal logic ($\GL$) in the style of Simpson, via proof theoretic techniques. 
We recover a labelled system, $\labIGL$, by restricting a non-wellfounded labelled system for $\GL$ to have only one formula on the right. The latter is obtained using techniques from cyclic proof theory, sidestepping the barrier that $\GL$'s usual frame condition (converse well-foundedness) is not first-order definable.
While existing intuitionistic versions of $\GL$ are typically defined over only the box (and not the diamond), our presentation includes both modalities.

Our main result is that $\labIGL$ coincides with a corresponding semantic condition in birelational semantics: the composition of the modal relation and the intuitionistic relation is conversely well-founded. We call the resulting logic $\IGL$. While the soundness direction is proved using standard ideas, the completeness direction is more complex and necessitates a detour through several intermediate characterisations of $\IGL$. 
\end{abstract}

\section{Introduction}\label{sec:intro}

\emph{G{\"o}del-L{\"o}b logic} ($\GL$) originates in the \emph{provability} reading of modal logic: $\Box$ is interpreted as ``it is provable that'', in an arithmetical theory with suitable coding capacity, inducing its corresponding \emph{provability logic}. 
 L{\"o}b formulated a set of necessary conditions on the provability logic of Peano Arithmetic ($\PA$), giving rise to $\GL$, extending basic modal logic $\K$ by what we now call \emph{L{\"o}b's axiom}:
$\Box(\Box A \limp A) \limp \Box A$. 
Somewhat astoundingly $\GL$ turns out to be \emph{complete} for 
$\PA$'s provability logic, a celebrated result of Solovay~\cite{Solovay76}: all that $\PA$ can prove about its own provability is already a consequence of a relatively simple (and indeed decidable) propositional modal logic.

Proof theoretically L\"ob's axiom represents a form of \emph{induction}.
Indeed, at the level of modal logic semantics, $\GL$ enjoys a correspondence with transitive relational structures that are \emph{terminating}\footnote{Other authors refer to this property as \emph{Noetherian} or \emph{conversely well-founded}, but we opt for this simpler nomenclature.}~\cite{Segerberg71}.
Duly, in computer science, L{\"o}b's axiom 
has inspired several variants of modal type theories that extend simply-typed lambda calculus with some form of \emph{recursion}.
These range, for instance, from the seminal work of Nakano~\cite{Nakano00}, to more recent explorations into guarded recursion~\cite{Birkedal12} and intensional recursion~\cite{Kavvos21}.
Based in type theories, these developments naturally cast $\GL$ in a \emph{constructive} setting, but
little attention was given
to analysing
the induced intuitionistic modal logics.
Indeed, for this reason, \cite{CloustonGore15} has proposed a more foundational basis for studying computational interpretations of $\GL$, by way of a sequent calculus for the logic of the topos of trees.

Returning to the provability reading of modal logic, several constructive variants of $\GL$ have been independently proposed (see \cite{Litak14,vdGiessen22} for overviews). 
An important such logic is $\iGL$ (and its variants) which now enjoys a rich mathematical theory, from semantics (e.g.,~\cite{Litak14}) to proof theory (e.g.,~\cite{vdGiessenIemhoff21, GoreShillito22}).
Interestingly, while $\iGL$ is known to be sound for the provability logic of Heyting Arithmetic ($\HA$), the intuitionistic counterpart of $\PA$, it is not complete. 
The provability logic of $\HA$ has been recently announced by~\cite{Mojtahedi22}, currently under review. 

One shortfall of all the above mentioned approaches is that they do not allow us to recover a bona fide computational interpretation of \emph{classical} $\GL$ along, say, the G\"odel-Gentzen negative translation (GG), a standard way of lifting interpretations of intuitionistic logics to their classical counterparts.
Indeed it was recently observed that the `$\mathsf i$' (or `constructive') traditions of intuitionistic modal logic are too weak to validate the GG translation \cite{DasMarin22,DasMarin23}.

On the other hand modal logic's relational semantics effectively renders it a fragment of usual first-order predicate logic (FOL), the so-called \emph{standard translation}. 
Interpreting this semantics in an intuitionistic meta-theory defines a logic that \emph{does} validate GG, for the same reason that intuitionistic FOL GG-interprets classical FOL.
This is the approach taken (and considerably developed) by Simpson \cite{Simpson94PhD}, building on earlier work of 
Fischer Servi \cite{FischerServi77} and Plotkin and Stirling \cite{plotkin:stirling:86tark}.
The resulting logic $\IK$ (and friends) enjoys a remarkably robust proof theory by way of \emph{labelled systems}, which may be duly seen as a fragment of Gentzen's systems for intuitionistic FOL by way of the standard translation.

\paragraph*{Contribution}
In this work we develop an intuitionistic version of $\GL$, following Simpson's methodology \cite{Simpson94PhD}.
In particular, while logics such as $\iGL$ are defined over only the $\Box$, our logic is naturally formulated over both the $\Box$ and the $\Diamond$. 
A key stumbling block to this end, as noticed already in the classical setting by Negri in \cite{Negri05}, is that $\GL$'s correspondence to terminating relations cannot seemingly be inlined within a standard labelled system: termination is not even FOL-definable.
To side-step this barrier we draw from a now well-developed proof-theoretic approach to (co)induction: \emph{non-wellfounded proofs} (e.g.\ \cite{NiwinskiWalukiewicz96,BrotherstonSimpson11,Simpson17,BaeldeDS16,DasP18,Das20}). 
Such proofs allow infinite branches, and so are (typically) equipped with a \emph{progress condition} that ensures sound reasoning.
Starting from a standard labelled system for transitive relations, $\Kfour$, we recover a labelled calculus $\labGL$ for $\GL$ by allowing non-wellfounded derivations under a typical progress condition.
In fact, our progress condition is precisely the one from Simpson's \emph{Cyclic Arithmetic} \cite{Simpson17,Das20}.
Following (the same) Simpson, we duly recover an intuitionistic version $\labIGL$ of $\GL$ \emph{syntactically} in a standard way: we restrict $\labGL$ to one formula on the right.

At the same time we can recover an intuitionistic version of $\GL$ \emph{semantically} by suitably adapting the \emph{birelational semantics} of intuitionistic modal logics, which combine the partial order $\leq$ of intuitionistic semantics with the accessibility relation $R$ of modal semantics. 
Our semantic formulation
$\birelIGL$ is duly obtained by reading the termination criterion of classical~$\GL$'s semantics intuitionistically: the composition $\leq ; R$ must be terminating. 
Our main result is that these two characterisations, $\labIGL$ and $\birelIGL$, are indeed equivalent, and we duly dub the resulting logic $\IGL$.

The soundness direction
is proved via standard techniques from intuitionistic modal logic and non-wellfounded proof theory.
The completeness direction, on the other hand, is more cumbersome. 
To this end we exercise an intricate combination of proof theoretic techniques, 
necessitating two further (and ultimately equivalent) characterisations of $\IGL$: semantically $\predIGL$, essentially a class of intuitionistic FOL structures, and syntactically $\multiintlabGL$, a \emph{multi-succedent} variant of $\labIGL$.
These formulations facilitate a countermodel construction from failed proof search, inspired by Takeuti's for $\LJ$ \cite{Takeuti75}.
Due to the non-wellfoundedness of proofs we employ a determinacy principle to organise the construction,  a standard technique in non-wellfounded and cyclic proof theory (e.g.\ \cite{NiwinskiWalukiewicz96}).
For the reduction to $\labIGL$ we devise a form of \emph{continuous cut-elimination}, building on more recent ideas in non-wellfounded proof theory, cf.~\cite{BaeldeDS16,DasP18}.
Our `grand tour' of results is visualised in \cref{fig:tour}, also indicating the organisation of this paper.
Due to space constraints, proofs are relegated to appendices.

\begin{figure}[t]
    \centering
    \vspace{1.6em}
    \begin{tikzpicture}
        \node (ILGL) at (0,2.2) {$\labIGL$};
        \node (mILGL) at (0,0) {$\multiintlabGL$};
        \node (pred) at (6,0) {$\predIGL$}; 
        \node (birel) at (6,2.2) {$\birelIGL$}; 
        \draw [->] (mILGL) to node [midway,left] {\shortstack{\cref{thm:completeness}\\(\cref{sec:cut-elim})}} (ILGL) ; 
        \draw [->] (ILGL) to node [above] {\shortstack{\cref{thm:soundness}\\ (\cref{sec:soundness})}} (birel);
        \draw[->] (birel) to node[right]{\shortstack{\cref{prop:birel-sat-implies-pred-sat}\\(\cref{sec:predmodels})}} (pred);
        \draw[->] (pred) to node[above]{\shortstack{\cref{thm:completeness-of-mlIGL}\\(\cref{sec:completeness})}} (mILGL);
    \end{tikzpicture}
    \caption{Summary of our main results, where arrows $\to$ denote inclusions of modal logics.
 }
    \label{fig:tour}
\end{figure}
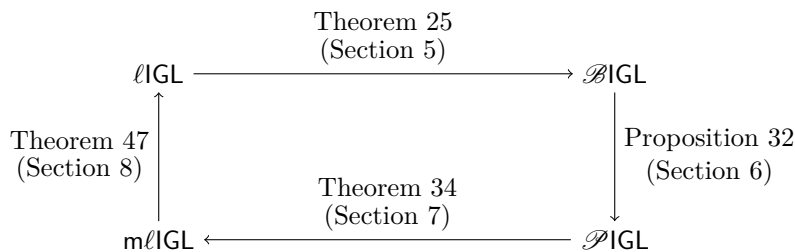

\paragraph*{Other related work}
The
proof theory of $\GL$ is (in)famously complex. 
The first sequent calculus of $\GL$ was considered in~\cite{Leivant81} but its cut-elimination property was only finally settled (positively) in~\cite{GoreRamanayake08} after several attempts~\cite{Valentini83,Sazaki01,Moen01}.
Intuitionistic versions of sequent calculi for~$\GL$ are developed in~\cite{vdGiessenIemhoff21, GoreShillito22} and provide calculi for $\iGL$.
Labelled calculi \cite{Negri05} and \emph{nested calculi}~\cite{Poggiolesi09} have also been proposed.
However all of these calculi are arguably unsatisfactory:
the modality introduction rules are non-standard (which is atypical for labelled calculi) and, in particular, modal rules may change the polarity of formula occurrences, due to the inductive nature of L\"ob's axiom. 
 
It is possible to recover a system for $\GL$ by admitting non-wellfounded proofs in a sequent calculus for $\Kfour$, as observed by Shamkanov~\cite{Shamkanov14}.
An intuitionistic version of this system has been studied by Iemhoff~\cite{Iemhoff16}.
In these works
both the base calculus and the corresponding correctness criterion are bespoke, rather than `recovered' from established foundations.

Non-wellfounded proofs 
originate in the study of modal logics with fixed points, in particular Niwinski and Walukiewicz's seminal work on the
$\mu$-calculus
~\cite{NiwinskiWalukiewicz96}.
These ideas were later inlined into the proof theory of FOL with certain inductive definitions by Brotherston and Simpson (e.g.\ \cite{BrotherstonSimpson11}), a source of inspiration for the present work.
As already mentioned, recent extensions of these ideas to $\PA$ \cite{Simpson17,Das20} and advances on cut-elimination \cite{BaeldeDS16,DasP18} are quite relevant to our development.

\section{Preliminaries on (classical) modal logic}\label{sec:prelims}
Throughout this work we work with a set $\Prop$ of \demph{propositional symbols}, written $p,q,$ etc., which we simultaneously construe as unary predicate symbols when working in predicate logic. 
For the latter we also assume a single binary relation symbol $R$ and a countable set $\Var$ of \demph{(individual) variables}, written $x,y,$ etc.

\subsection{Language and semantics}
\label{subsec:language-and-semantics}

The formulas of modal logic are generated by the following grammar:
\[A \quad \coloncolonequals \quad  p \in \Prop \ \mid\  \bot \ \mid\  A \land A \ \mid\  A \lor A \ \mid\  A \limp A \ \mid\  \Box A \ \mid\  \Diamond A\]
As usual we write $\lnot A := A \limp \bot$, and employ standard bracketing conventions.

Modal formulas are interpreted over \demph{relational frames} $\Fframe=(W,R^\Fframe)$ formed of a non empty set of \demph{worlds} $W$ equipped with an \demph{accessibility relation} $R^\Fframe \subseteq W \times W$.\footnote{We parameterise the relations by the frame or models to be able distinguish it from the fixed binary relation symbol $R$.}
 A \demph{relational model} $\Mmod = (W,R^\M,V)$ is a structure where $(W,R^\M)$ is a frame with a \demph{valuation} $V: W \to \mathcal{P}(\Prop)$.
 Let~$\Mmod = (W,R^\M,V)$ be a relational model.
 For worlds $w\in W$ and formulas $A$ we define the satisfaction judgement $\Mmod, w \models A$ as follows:

\smallskip

\begin{minipage}[t]{.35\textwidth}
	\begin{itemize}
		\item $\Mmod,w \models p$ if $p \in V(w)$; \ 
		\item $\Mmod,w \notmodels \bot$; \ 
	\end{itemize}
\end{minipage}
\begin{minipage}[t]{.55\textwidth}
	\begin{itemize}
		\item $\Mmod,w \models A \land B$ if $\Mmod,w \models A$ and $\Mmod,w \models B$; \ 
		\item $\Mmod,w \models A \lor B$ if $\Mmod,w \models A$ or $\Mmod,w \models B$; \ 
	\end{itemize}
\end{minipage}

\smallskip

\begin{minipage}[t]{.9\textwidth}
	\begin{itemize}
		\item $\Mmod,w \models A \limp B$ if $\Mmod,w \models A$ then $\Mmod,w \models B$; \ 
		\item $\Mmod,w \models \Box A$ if for all $v$ such that $wR^\M v$ we have $\Mmod,v \models A$; \ 
		\item $\Mmod,w \models \Diamond A$ if there exists $v$ such that $wR^\M v$ and $\Mmod,v \models A$; \ 
	\end{itemize}
	\end{minipage}

\smallskip
\noindent
If $\M,w \models A$ for all $w \in W$ we simply write $\M \models A$, and if $(W,R^\M,V) \models A$ for all valuations $V$ based on frame $\Fframe = (W, R^\M)$, we simply write $\Fframe\models A$.

\subsection{Axiomatisations and G\"odel-L\"ob logic}

Turning to syntax, let us now build up the modal logics we are concerned with axiomatically, before relating them to the semantics just discussed. 
The modal logic $\K$ is axiomatised by all the theorems of classical propositional logic (CPL) together with axiom $\kaxiom$ 
and closed under the rules $\sf (mp)$ (\emph{modus ponens}) and $\sf nec$ (\emph{necessitation}) from \cref{fig:axiomatisations}.
The logic $\Kfour$ is defined in the same way by further including the axiom $\fouraxiom$ from \cref{fig:axiomatisations}.

Referring back to the earlier semantics,
the following characterisation is well-known~\cite{Blackburn01Book}:

\begin{theorem}
[$\K$/$\Kfour$ characterisation]
\label{thm:correspondence-k-kfour}
$\K \proves A$ (resp.~$\Kfour \proves A$) iff $A$ is satisfied in all relational frames (resp.~with transitive accessibility relation). 
\end{theorem}

\begin{figure}[t]
    \centering
    \begin{align*}
         \small
    &\kaxiom : \Box(A \limp B) \limp (\Box A \limp \Box B)
    &&\vliinf{\sf (mp)}{}{B}{A \limp B}{A}
    &\vlinf{\sf (nec)}{}{\Box A}{A}\\
    &\fouraxiom : \Box A \limp \Box \Box A
    &&\wlobaxiom :   \Box(\Box A \limp A) \limp \Box A
    &
    \end{align*}
    \caption{Some modal axioms and rules.}
\label{fig:axiomatisations}
\end{figure}

The main subject matter of this work is an extension of $\K$ by the L\"ob axiom $\wlobaxiom$, in \cref{fig:axiomatisations}, a sort of induction principle: 

\begin{definition}[G{\"o}del-L{\"o}b logic]
\label{def:GodelLob}
$\GL$ is defined by extending $\Kfour$ by the \demph{L\"ob axiom} $\wlobaxiom$
and is closed under $\sf (mp)$ and $\sf (nec)$.\footnote{Alternatively, $\GL$ can be axiomatised by adding the $\wlobaxiom$ axiom to $\K$, as $\fouraxiom$ can be derived from $\wlobaxiom$.}
\end{definition}
Semantically $\wlobaxiom$ says that, as long as worlds satisfy $A$ whenever all its successors satisfy $A$, then $A$ holds universally. 
This amounts to a `reverse' induction principle for the accessibility relation.
Indeed we have an associated characterisation for $\GL$ just like that for $\Kfour$.

We call a frame is \demph{terminating} if its accessibility relation has no infinite path.
\begin{theorem}
[$\GL$ characterisation]
\label{thm:gl-correspondence}
    $\GL \proves A$ iff all transitive terminating frames satisfy $A$.
\end{theorem}

\subsection{Labelled calculi and the standard translation}
\label{subsec:labelled-calc}

The relational semantics of modal logic may be viewed as a bona fide fragment of predicate logic. 
Recalling the predicate language we fixed at the start of the section,
the \demph{standard translation} is defined as follows: for individual variables $x$ and modal formulas $A$ we define the predicate formula $\sttr x A$ by:

\smallskip

\begin{minipage}[t]{.55\textwidth}
	\begin{itemize}
		\item $\sttr x p$ is $p(x)$; \ 
		\item $\sttr x \bot$ is $\bot$; \ 
		\item $\sttr x A\star B$ is $(\sttr x A)\star (\sttr x B)$ for $\star \in \{\lor, \land, \limp\}$; \ 
	\end{itemize}
\end{minipage}
\begin{minipage}[t]{.35\textwidth}
	\begin{itemize}
		\item $\sttr x \Diamond A$ is $\exists y (xRy \land (\sttr y A)) $ ; \ 
		\item $\sttr x \Box A$ is $\forall y (xRy \limp (\sttr y A))$.
	\end{itemize}
\end{minipage}

\medskip

This induces a well-behaved proof theory as a fragment of usual first-order predicate systems~\cite{Negri05}, adaptable to many extensions under correspondence theorems such as \cref{thm:correspondence-k-kfour}.

A \demph{(labelled) sequent} is an expression $\relats R, \Gamma \seqar \Delta$, where $\relats R$ is a set of \demph{relational atoms}, i.e.\ formulas of form $xRy$, and $\Gamma$ and $\Delta$ are multisets of \demph{labelled formulas}, i.e.\ formulas of form $x:A$. 
We sometimes refer to the variable $x$ in $x:A$ as a \demph{label}. 
We write $\Lab(S)$ for the subset of labels/variables that occur in a given sequent $S$. Similarly so for $\Lab(\relats R)$, etc.

Notationally, we have identified labelled formulas with the standard translation at the beginning of this section.
This is entirely suggestive, as we can now easily distil systems for modal logics of interest by appealing to the sequent calculus $\LK$ for first-order predicate logic.
In this vein, the labelled calculus $\labK$ for modal logic $\K$ is given in \Cref{fig:labelled-k}. 
From here, under the aforementioned correspondence between $\Kfour$ and transitive frames, we have a system $\labKfour$ for $\Kfour$ that extends $\labK$ by the \demph{transitivity rule}:
    \[
    \vlinf{\trans}{}{\relats R, xRy,yRz , \Gamma \seqar \Delta}{\relats R, 
    xRy,yRz,
    xRz, \Gamma \seqar \Delta}
    \]

For a labelled system $\system$ we write $\system \proves \relats R, \Gamma \seqar \Delta$, if there is a proof of $\relats R, \Gamma \seqar \Delta$ using the rules from $\system$. We write $\system \proves x : A$, or even $\system \proves A$, to mean $\system \proves \emptyset \seqar x : A$.
Almost immediately from \cref{thm:correspondence-k-kfour} and metatheorems for predicate logic, we have:

\begin{proposition}
[Soundness and completeness]\label{prop:sound-compl-labIK}
    $\labK \proves x:A$ (resp.~$\labKfour \proves x:A$) iff $\K \proves A$ (resp.~$\Kfour \proves A$).
\end{proposition}

\begin{figure}[t]
    \emph{Identity and cut:}
\[
\small
\vlinf{\id}{}{\relats R, x:p \seqar  x:p}{}
\qquad
\vliinf{\cut}{}{\relats R, \Gamma, \Gamma'\seqar \Delta, \Delta'}{\relats R, \Gamma \seqar \Delta, x:A}{\relats R, \Gamma', x:A \seqar \Delta'}
\]

\smallskip

\begin{minipage}[t]{.6\textwidth}
    \emph{Structural rules:}
\[
\small
\begin{array}{cc}
    \vlinf{\leftrule \wk}{}{\relats R, \Gamma, x:A \seqar \Delta}{\relats R, \Gamma \seqar \Delta} 
    &
     \vlinf{\leftrule \contr}{}{\relats R, \Gamma, x:A \seqar \Delta}{\relats R, \Gamma, x:A, x:A \seqar \Delta}
     \\
     \noalign{\smallskip}
    \vlinf{\rightrule \wk}{}{\relats R, \Gamma \seqar \Delta, x:A}{\relats R, \Gamma \seqar \Delta} 
     &
     \vlinf{\rightrule \contr}{}{\relats R, \Gamma \seqar \Delta, x:A}{\relats R, \Gamma \seqar \Delta, x:A, x:A}
\end{array}
\]
\end{minipage}
\begin{minipage}[t]{.35\textwidth}
    \emph{Relational structural rule:}
\[
\small
\vlinf{\thinning}{}{\relats R, \relats R', \Gamma \seqar \Delta}{\relats R, \Gamma \seqar \Delta}
\]
\end{minipage}

\smallskip

\emph{Propositional logical rules:}
\[
\small
\vlinf{\leftrule \bot}{}{\relats R, x:\bot, \Gamma \seqar \Delta }{}
\qquad
\text{\color{black!40} (no right rule for $\bot$)}
\]
\[
\small
\vliinf{\leftrule \limp}{}{\relats R, \Gamma,\Gamma' , x:A \limp B \seqar \Delta, \Delta'}{\relats R, \Gamma \seqar \Delta, x:A}{\relats R, \Gamma' , x:B \seqar \Delta'}
\qquad
\vlinf{\rightrule\limp}{}{\relats R, \Gamma \seqar \Delta, x:A\limp B}{\relats R, \Gamma, x:A \seqar \Delta, x:B}
\]
\[
\small
\begin{array}{cc}
    \vlinf{\leftrule\land}{i\in \{0,1\}}{\relats R , \Gamma, x:A_0 \land A_1 \seqar \Delta }{\relats R, \Gamma, x:A_i \seqar \Delta} &  
     \vliinf{\leftrule \lor}{}{\relats R, \Gamma, x:A\lor B \seqar \Delta}{\relats R, \Gamma, x:A \seqar \Delta}{\relats R, \Gamma, x:B \seqar \Delta}
    \\
    \noalign{\smallskip}
        \vlinf{\rightrule\lor}{i \in \{0,1\}}{\relats R, \Gamma \seqar \Delta, x:A_0\lor A_1}{\relats R,\Gamma \seqar \Delta, x:A_i}
&          
\vliinf{\rightrule\land}{}{\relats R, \Gamma \seqar \Delta,  x:A\land B}{\relats R, \Gamma \seqar \Delta, x:A}{\relats R, \Gamma \seqar \Delta, x:B}
\end{array}
\]

\smallskip

\emph{Modal logical rules:}
\[
\small
\begin{array}{cc}
    \vlinf{\leftrule \Diamond}{\text{$y $ fresh}}{\relats R, \Gamma , x:\Diamond A\seqar \Delta}{\relats R, xRy, \Gamma, y:A \seqar \Delta} &  
    \vlinf{\rightrule \Diamond}{}{\relats R, xRy, \Gamma \seqar \Delta , x:\Diamond A }{\relats R, xRy, \Gamma \seqar \Delta, y:A}
    \\
    \noalign{\smallskip}
    \vlinf{\rightrule \Box}{\text{$y$ fresh}}{\relats R, \Gamma \seqar \Delta, x:\Box A}{\relats R, xRy, \Gamma \seqar \Delta, y:A}
     & 
    \vlinf{\leftrule \Box}{}{\relats R , x Ry, \Gamma, x:\Box A  \seqar \Delta}{\relats R, xRy,  \Gamma, y:A  \seqar \Delta}
\end{array}
\]
    \caption{The standard labelled calculus $\labK$ for modal logic $\K$. 
    }
    \label{fig:labelled-k}
\end{figure}

Naturally systems for many other modal logics can be readily obtained, when they correspond to simple frame properties, by adding further relational (structural) rules~\cite{Negri05}.
Indeed, let us call 
a labelled calculus \emph{standard} if it does not extend $\labK$ by any new logical or (non-relational) structural rules, nor any new logical axioms. 
This terminology is suggestive, since it forces the left and right introduction rules for modalities to coincide with those induced by the \emph{standard} translation.
As remarked by Negri in~\cite{Negri05}, typical standard calculi cannot be complete $\GL$, as termination of a relation is not even first-order definable.
We shall sidestep this barrier in the next section by making use of \emph{non-wellfounded} systems.

\section{Recovering a proof theoretic account for \GL}\label{sec:gl}

In another branch of the proof theory literature, structural treatments of induction and well-foundedness have been developed in the guise of \emph{non-wellfounded} and \emph{cyclic} proofs, e.g.~\cite{NiwinskiWalukiewicz96,BrotherstonSimpson11,BaeldeDS16,BerardiTatsuta17,Simpson17,Das20}. 
Here non-wellfoundedness in proofs allows for inductive reasoning, and soundness is ensured by some global correctness condition.
By incorporating these ideas into modal proof theory, one can design a non-wellfounded proof system for $\GL$~\cite{Shamkanov14}.
In this section we recover a \emph{standard} labelled calculus for $\GL$, in the sense of the preceding discussion.
Our presentation is based on the correctness condition for \emph{Cyclic Arithmetic} in~\cite{Simpson17,Das20}.

\subsection{A standard calculus for \GL, via non-wellfounded proofs}
In what follows, we consider systems $\system$ that will typically be some fragment of the system $\labKfour$ given earlier.
The definitions we give apply to all `non-wellfounded' systems of this work.

\begin{definition}
    [Preproofs]
    A \demph{preproof} in a system $\system$ is a possibly infinite derivation generated from the rules of $\system$.
\end{definition}

As preproofs may, in particular, be non-wellfounded, they may conclude fallacious theorems, so we require a correctness criterion.
Appealing to the correspondence of $\GL$ over transitive terminating relations, we may directly import a `trace' condition first used in \cite{Simpson17}.

\begin{definition}
    [Traces and proofs]
    Fix a preproof $P$ and an infinite branch $(S_i)_{i<\omega}$.
    A \demph{trace} along $(S_i)_{i<\omega}$ is a sequence of variables $(x_i)_{i<\omega}$ such that either:

    \noindent
    \begin{enumerate*}
        \item $x_{i+1} = x_i$; or, \ 
        \item\label{item:progress} $x_iRx_{i+1}$ appears in $S_i$ (then $i$ is a \demph{progress point} of $(x_i)_{i<\omega}$).
    \end{enumerate*}

    \noindent
    A trace is \demph{progressing} if it is not ultimately constant, i.e.\ if \ref{item:progress} above applies infinitely often.
    A branch $(S_i)_{i<\omega}$ is {progressing} if it has a progressing trace, and a preproof $P$ is {progressing}, or simply is a \demph{$\infty$-proof}, if each infinite branch has a progressing trace.

    \noindent
    We write $\system \nwfproves \, \relats R, \Gamma \seqar \Delta$ if there is a $\infty$-proof in $\system$ of the sequent $\relats R,\Gamma \seqar \Delta$.
\end{definition}

\begin{example}
[L\"ob and contra-L\"ob]
        \label{ex:loeb}
An example of a $\infty$-proof of $\wlobaxiom$ in $\labKfour$ is given in \cref{fig:infty-proofs-lkfour}, left.
Here we have used bullets $\bullet$ to identify identical subproofs, up to the indicated renamings of variables.
The preproof is indeed progressing: it has only one infinite branch, whose {progress points} are coloured red and trigger at each iteration of the $\bullet$-loop.
Notice that each sequent has only one formula on the right-hand side (RHS).

Another example of a $\infty$-proof in $\labKfour$ is the \emph{contraposition} of $\wlobaxiom$, given in \cref{fig:infty-proofs-lkfour}, right.
We have merged several steps, and omitted some routine structural steps and initial sequents, a convention we shall continue to employ throughout this work.
Again the preproof is indeed progressing, by the same argument as before.
Notice, this time, that there are two formulas in the RHS of the premiss of $\bullet$.
Indeed, by basic inspection of proof search, there is no cut-free $\infty$-proof avoiding this feature.
\end{example}
\begin{figure}
    \centering
     \[
    \small
    \vlderivation{
    \vlin{\rightrule \limp}{}{\seqar x: \Box (\Box p \limp p) \limp \Box p }{
    \vlin{\rightrule \Box}{}{x: \Box (\Box p \limp p) \seqar x:\Box p}{
    \vlin{\leftrule \contr}{\bullet}{\red{xRy} , x:\Box (\Box p \limp p) \seqar y:p}{
    \vlin{\leftrule \Box}{}{xRy, x:\Box (\Box p \limp p),x:\Box (\Box p \limp p)\seqar y:p}{
    \vlin{\leftrule \limp}{}{xRy, x:\Box (\Box p \limp p), y:\Box p \limp p\seqar y:p}{
        \vlin{\rightrule \Box}{}{xRy, x: \Box (\Box p \limp p) \seqar y:\Box p}{
        \vlin{\trans}{}{xRy,\red{yRz},x:\Box (\Box p \limp p) \seqar z:p}{
        \vlin{\leftrule \contr}{\bullet[z/y]}{xRz,x:\Box (\Box p \limp p) \seqar z:p}{\vlhy{\vdots}}
        }
        }
    }
    }
    }
    }
    }
    }
    \quad
    \vlderivation{
    \vlin{\rightrule\limp}{}{\seqar x:\Diamond p \limp \Diamond (p \land \Box \lnot p)}{
    \vlin{\leftrule\Diamond}{}{x:\Diamond p \seqar \Diamond (p \land \Box \lnot p)}{
    \vlin{\rightrule \Diamond }{\bullet}{\red{xRy}, y:p \seqar x: \Diamond (p \land \Box \lnot p)}{
    \vlin{\rightrule \land}{}{xRy, y:p \seqar x: \Diamond (p \land \Box \lnot p), y:p \land \Box \lnot p}
    {
        \vlin{\rightrule \lnot, \rightrule \Box}{}{xRy, y:p \seqar x: \Diamond (p \land \Box \lnot p), y:\Box \lnot p }{
        \vlin{\trans}{}{xRy, \red{yRz}, y:p, z:p \seqar x:\Diamond (p \land \Box \lnot p)}{
        \vlin{\rightrule \Diamond}{\bullet[z/y]}{xRz, z:p \seqar x:\Diamond (p \land \Box \lnot p)}{\vlhy{\vdots}}
        }
        }
    }
    }
    }
    }
    }
    \]
    \caption{An $\infty$-proof in $\labKfour$ of L\"ob's axiom, left, and its contraposition, right. Identity steps on $y:p$ above the $\leftrule\limp$ step, left, and the $\rightrule \land$ step, right, are omitted for space considerations.}
    \label{fig:infty-proofs-lkfour}
\end{figure}

\begin{remark}
    [On regularity]
    In non-wellfounded  proof theory, special attention is often paid to the
    subset of \emph{regular} preproofs, which may be written as finite (possibly cyclic) graphs, as in \cref{ex:loeb} above. 
Nonetheless these will play no role in the present work, as we are purely concerned with logical and proof theoretic investigations, not with effectivity.
\end{remark}
 
The progress condition is invariant under expansion of relational contexts in a preproof:
We may omit consideration of the thinning rule $\thinning$ (\cref{fig:labelled-k}) when reasoning about $\infty$-proofs:

\begin{observation}
    \label{obs:thinning}
    $\thinning$ is eliminable $\infty$-proofs of $\system$.
\end{observation}

\subsection{Soundness and completeness}
Let us now argue that our notion of $\infty$-proof for $\labKfour$ is sound and complete for $\GL$.
The most interesting part is soundness, comprising a contradiction argument by infinite descent as is common in non-wellfounded proof theory, relying on the characterisation result of \cref{thm:gl-correspondence}:

\begin{proposition}
[Soundness]
\label{prop:soundness-of-labGL}
If $\labKfour \nwfproves x:A$ then $\GL \proves A$.    
\end{proposition}

On the other hand, thanks to \cref{ex:loeb} and the known completeness of $\labKfour$ for $\Kfour$, we can use $\cut$-rules to derive:
\begin{proposition}
    [Completeness]
    \label{prop:completeness-labgl}
    If $\GL \proves A$ then $\labKfour \nwfproves x:A$.
\end{proposition}

These results motivate the following notation:
\begin{definition}
    $\labGL$ is the class of $\infty$-proofs of $\labKfour$.
\end{definition}

Henceforth for a class of $\infty$-proofs $\nwfprfclass$ we may write simply $\nwfprfclass \proves S$ if $\nwfprfclass$ contains a $\infty$-proof of the sequent $S$.
For instance, writing $\labGL \proves S$ is the same as $\labKfour \nwfproves S$.

\section{Recovering intuitionistic versions of \GL\ from syntax and semantics}\label{sec:igl}

In this section we propose two intuitionistic versions of $\GL$, in the style of Simpson, respectively by consideration of the proof theory and semantics of $\GL$ discussed in the previous sections. Later sections are then devoted to proving the equivalence of these two notions. 

\subsection{An intuitionistic \GL, via syntax}

Following Gentzen, it is natural to define intuitionistic calculi based on their classical counterparts by restricting sequents to one formula on the RHS. 
However, when implementing this restriction to different starting calculi for modal logic $\K$ (based on sequents, nested sequents, labelled sequents) one ends up with different intuitionistic variants of $\K$~(see e.g. \cite{DasMarin22}).
The restriction of the ordinary sequent calculus for $\K$ defines a logic which is not compatible with the standard translation in the way we described for $\K$.
On the other hand labelled calculi are well designed for this purpose.

\begin{definition}
    The system $\intlabK$ (resp.~$\intlabKfour$) is the restriction of $\labK$ (resp.~$\labKfour$) to sequents in which exactly one formula occurs on the RHS.
\end{definition}

Note that the rules $\rightrule \wk$ and $\rightrule \contr$ cannot be used in $\intlabKfour$, by the singleton restriction on the RHS of a sequent.
In the style of Simpson~\cite{Simpson94PhD}, 
we can from here duly recover a standard intuitionistic analogue of $\GL$, by restricting $\labGL$ to $\infty$-proofs with only singleton RHSs.

\begin{definition}
    We write $\intlabGL$ for the class of $\infty$-proofs of $\intlabKfour$.
\end{definition}

\begin{example}
    [L\"ob, revisited]
    \label{ex:loeb-igl}
    Recalling \cref{ex:loeb} earlier, note that the $\infty $-proof of L\"ob's axiom in \cref{fig:infty-proofs-lkfour}, left, indeed satisfies the singleton RHS restriction, and so $\labIGL \proves \wlobaxiom$.
    On the other hand its contraposition, right, does not satisfy this restriction. 
    We will see shortly that it is indeed \emph{not} a theorem of $\labIGL$. 
\end{example}

\subsection{An intuitionistic \GL, via semantics}

To give our semantic version of intuitionistic $\GL$, we must first recall models of intuitionistic modal logic.
\emph{Birelational semantics} include models $\Bmod$ with two relations, the intuitionistic~$\leq$ and the modal~$R^\Bmod$ (parameterised with $\Bmod$ to distinguish from the fixed predicate symbol $R$ used in this paper), which duly gives it the capacity to model intuitionistic modal logics. 
\begin{definition}[Birelational semantics]
\label{dfn:birel-semantics}
    A \demph{birelational frame} $\Fframe$ is a triple~$(W,\leq,R^\Fframe)$, where~$W$ is a nonempty set of \demph{worlds} equipped with a partial order~$\leq$ and an \demph{accessibility relation}~$R^\Fframe \subseteq W \times W$. We require the following frame conditions:
    \begin{itemize}
        \item[\normalfont (F1)] If $w\leq w'$ and $wR^\Fframe v$, then there exists $v'$ such that $v \leq v'$ and $w'R^\Fframe v'$.
        
        \item[\normalfont (F2)] If $wR^\Fframe v$ and $v \leq v'$, then there exists $w'$ such that $w \leq w'$ and $w'R^\Fframe v'$.
    \end{itemize}    
    \noindent
    A \demph{birelational model} is a tuple~$(W,\leq,R^\Bmod,V)$, where~$(W,\leq,R^\Bmod)$ is a birelational frame and~$V$ is a \demph{valuation}~$ W \to \calP(\Prop)$ that is monotone in~$\leq$, i.e., $w \leq w'$ implies $V(w) \subseteq V(w')$.

    Let~$\Bmod = (W,\leq,R^\Bmod,V)$ be a birelational model.
    For worlds $w\in W$ and formulas $A$ we define the satisfaction judgement $\Bmod, w \models A$ as follows:

\smallskip

\begin{minipage}[t]{.4\textwidth}
    \begin{itemize}
    \item $\Bmod,w \models p$ if $p \in V(w)$; \ 
    \item $\Bmod,w \notmodels \bot$; \ 
\end{itemize}
\end{minipage}
\begin{minipage}[t]{.5\textwidth}
\begin{itemize}
    \item $\Bmod,w \models A \land B$ if $\Bmod,w \models A$ and $\Bmod,w \models B$; \ 
    \item $\Bmod,w \models A \lor B$ if $\Bmod,w \models A$ or $\Bmod,w \models B$; \ 
\end{itemize}
\end{minipage}

\begin{minipage}[t]{.9\textwidth}
\begin{itemize}
    \item $\Bmod,w \models A \limp B$ if for all $w' \geq w$, if $\Bmod,w' \models A$ then $\Bmod,w' \models B$; \ 
    \item $\Bmod,w \models \Box A$ if for all $w' \geq w$ and for all $v$ such that $w'R^\Bmod v$ we have $\Bmod,v \models A$; \ 
    \item $\Bmod,w \models \Diamond A$ if there exists $v$ such that $wR^\Bmod v$ and $\Bmod,v \models A$; \ 
\end{itemize}
\end{minipage}
\medskip

\noindent
We write $\Bmod \models A$ if $\Bmod, w \models A$ for all $w \in W$.
\end{definition} 

\begin{lemma}[Monotonicity lemma, \cite{Simpson94PhD}]
\label{lem:monotonicity}
    Let $\Bmod = (W,\leq,R^\Bmod,V)$ be a birelational model. For any formula $A$ and $w,w' \in W$, if $w \leq w'$ and $\Bmod, w \models A$, then $\Bmod, w' \models A$.
\end{lemma}

A soundness and completeness theorem is recovered in~\cite{Simpson94PhD}, similarly to the classical case (\cref{prop:sound-compl-labIK} and \cref{thm:correspondence-k-kfour}).

\begin{theorem}
[\cite{Simpson94PhD}]
    $\labIK \proves x : A$ (resp.~$\labIKfour \proves x : A$) iff $A$ is satisfied in all birelational models (resp.~with transitive accessibility relation) 
\end{theorem}

We want to introduce a birelational counterpart of the transitive and terminating models of $\GL$. 
The class of models that we will introduce here is different from the birelational semantics known for logic $\iGL$ \cite{Litak14,ArdeshirMojtahedi18}
which only require termination of the accessibility relation.
Due to the \emph{global} nature of $\Box$-evaluation in 
\cref{dfn:birel-semantics}, we here
require the \emph{composition} of $\leq$ and $R^\Bmod$ to be terminating.

\begin{definition}
    $\birelIGL$ is the class of birelational models~$\Bmod = (W,\leq,R^\Bmod,V)$ such that:
    \begin{itemize}
        \item $R^\Bmod$ is {transitive}; and,
        \item $(\leq ; R^\Bmod)$ is {terminating}, i.e., there are no infinite paths $x_1 \leq y_1 R^\Bmod x_2 \leq y_2 R^\Bmod x_3 \dots$.
    \end{itemize}
    For a formula $A$, we write $\birelIGL \models A$ to mean that $\Bmod \models A$ for all $\Bmod \in \birelIGL$.
\end{definition}

\begin{example}
[Contra-L\"ob, revisited]
\label{ex:contraloeb-igl}
    Recalling \cref{ex:loeb,ex:loeb-igl} 
    we indeed have that the contraposition of L\"ob's axiom, $\Diamond p \limp \Diamond(p \wedge \Box \neg p)$, is not valid in $\birelIGL$. 
    It is falsified at
    world $w_1$ of the following $\birelIGL$-model $\Bmod$ where we assign $p$ to all worlds. In the picture we omit transitive (and reflexive) edges of $R^\Bmod$ (and $\leq$).
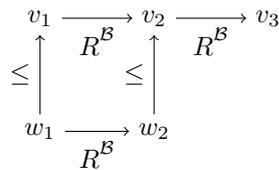
    \begin{center}
        \begin{tikzpicture}
        \node (x1) at (0,0) {$w_1$};
        \node (x2) at (1.5,0) {$w_2$};
        \node (y1) at (0,1.5) {$v_1$}; 
        \node (y2) at (1.5,1.5) {$v_2$}; 
         \node (y3) at (3,1.5) {$v_3$}; 
        \draw [->] (x1) to node [midway,below] {$R^\Bmod$} (x2) ; 
        \draw [->] (y1) to node [midway,below] {$R^\Bmod$} (y2) ; 
        \draw [->] (y2) to node [midway,below] {$R^\Bmod$} (y3) ; 
        \draw [->] (x1) to node [midway,left] {$\leq$} (y1) ; 
        \draw [->] (x2) to node [midway,left] {$\leq$} (y2) ; 
    \end{tikzpicture}
    \end{center}    
\end{example}

\section{Soundness}
\label{sec:soundness}

The main result of this section is the soundness of $\labIGL$ with respect to $\birelIGL$ (\cref{thm:soundness,cor:soundness-specialised-to-theorems}). 
First we must extend the notion of satisfaction in birelational models to labelled sequents.

Let us fix a sequent $S = (\relats R, \Gamma \seqar x:A)$ and a birelational model $\Bmod = (W,\leq,R^\Bmod,V)$ for the remainder of this section.
    An \demph{interpretation} of $S $ into $\Bmod $ is a function~$\calI : \Lab(S) \to W$ such that $\calI(x)R^\Bmod\calI(y)$ whenever $xRy \in \relats R$.
    We write
    \[
    \Bmod, \calI \models \relats R , \Gamma \seqar x : A 
    \quad \text{ if } \quad 
    \Bmod, \calI(x) \models A \text{ when }
    \Bmod, \calI(y) \models B \text{ for all } y : B \in \Gamma.
    \]
    If $\Bmod, \calI \models S$ for all interpretations~$\calI$ of~$S$ into~$\Bmod$, we simply write $\Bmod \models S$, 
    and if $\Bmod \models S$ for all models~$\Bmod \in \birelIGL$, we simply write $\birelIGL \models S$.
See Appendix~\ref{app-soundness} for a proof of the following.

\begin{observation}
\label{lem:models-form-seq}
    $\Bmod \models \emptyset , x:B_1, \ldots, x:B_n \seqar x : A$ iff $\Bmod \models (B_1 \land \ldots \land B_n) \limp A$.
\end{observation}

In order to prove local soundness of $\labIKfour$ rules, we use a lifting lemma similarly to the one in~\cite{Simpson94PhD} whose proof relies on the tree-like structure of $\relats R$. 

\begin{definition}[(Quasi-)tree-like]
    $\relats R$ is a \demph{tree} if there is $x_0 \in \Lab(\relats R)$ 
    such that for each $x \in \Lab(\relats R)$, $x\not=x_0$,  
    there is a unique sequence 
    $x_0Rx_1, x_1Rx_2, \dots, x_mRx \in \relats R$. 
    $\relats R$ is a \demph{quasi-tree} if 
    there is some 
    $\relats R_0 \subseteq \relats R \subseteq \CL{\relats R_0}$ where $\relats R_0$ is a tree and $\CL{\relats R_0}$ denotes the transitive closure of $\relats R_0$.
    Sequent $S$ is \demph{(quasi-)tree-like} if either $\relats R = \emptyset$ and $\Lab(\Gamma) \subseteq \{x\}$, or $\relats R$ is a (resp., quasi-)tree and $\Lab(\Gamma) \cup \{ x\} \subseteq \Lab(\relats R)$.
\end{definition}

\begin{lemma}[Lifting lemma]\label{lem:liftinglemma}       
    Suppose $S$ is quasi-tree-like. 
    Let $\calI$ be interpretation of~$S$ into $\Bmod$, 
    $x \in \Lab(S)$ and~$w \geq \calI(x)$.
    There is an interpretation~$\calI'$ of~$S$ into~$\Bmod$ such that $\calI'(x) = w$ and for all~$y \in \Lab(S)$ we have~$\calI'(y) \geq \calI(y)$.
\end{lemma}

Let us employ some conventions on $\infty$-proofs.
Note that for every inference rule of $\labIKfour$, 
except for $\thinning$ and $\cut$, the premiss(es) 
are quasi-tree-like whenever the conclusion is.
By \cref{obs:thinning} we shall duly assume that~$\thinning$ is not used.
Note that this forces relational contexts to be \emph{growing}, bottom-up: the relational context of a sequent always contains those below it. 
If a cut-formula $z:C$ has label $z$ that does not occur in the conclusion, we may safely rename $z$ to a variable that does. 
Thus we may assume that any $\infty$-proof in $\labIGL$ of $x : A$ has only quasi-tree-like sequents in it. 

\begin{proposition}[Local soundness]\label{prop:localsoundness}
    Suppose $S$ is quasi-tree-like
    Let $\calI$ be an interpretation of $S$ into $\Bmod$ such that $\Bmod, \calI \notmodels S$. 
    For any inference step 
    of $\labIKfour\setminus\{ \thinning\}$ that $S$ concludes,   
    there is a premiss $S'$ 
    and an interpretation $\calI'$ of $S'$ into $\Bmod$ such that $\Bmod, \calI' \notmodels S'$ and $\calI'(z) \geq \calI(z)$ for all $z \in \Lab(S)$.
\end{proposition}
    The proof uses some direct calculations in most cases and the lifting lemma (\cref{lem:liftinglemma})
    to handle rules $\rightrule \limp$ and $\rightrule \Box$.
    We give some cases in \cref{app-soundness}.
    From here, as in the classical setting for $\labGL$, we can employ a contradiction argument by infinite descent to conclude:

\begin{theorem}[Soundness]\label{thm:soundness}
    Suppose $S$ is quasi-tree-like. Then
    $\labIGL \proves S$ implies~$\birelIGL \models S$.
\end{theorem}

\begin{corollary}
\label{cor:soundness-specialised-to-theorems}
    If $\labIGL \proves x : A$ then $\birelIGL \models A$.
\end{corollary}

The remainder of this work is devoted to proving the converse result.
The sections that follow structure the proof into the three parts according to the arrows indicated in \cref{fig:tour}.

\section{From birelational models to Kripke predicate models}
\label{sec:predmodels}
Towards our countermodel construction in the next section we turn to \emph{predicate models}, essentially via the standard translation, whose internal structure is richer than that of birelational models, thus providing useful invariants for the sequel. 

\begin{definition}[Predicate models]
    A \demph{Kripke structure} is a tuple 
    \begin{itemize}
        \item $W$ is a non-empty set of \emph{worlds} partially ordered by $\leq$;
        \item $\{D_w\}_{w \in W}$ is a family of non-empty \demph{domains}, such that $D_w \subseteq D_{w'}$ whenever $w \leq w'$;
        \item $\{\Prop_w\}_{w\in W}$ is a family of mappings $\Prop_w : \Prop \to \mathcal P (D_w)$ such that for each $p \in \Prop$, $\Prop_w(p) \subseteq \Prop_{w'}(p)$ whenever $w \leq w'$;
        \item $\{ R_w \}_{w \in W}$ is a family of relations $R_w \subseteq D_w \times D_w$ such that $R_w \subseteq R_{w'}$ whenever $w \leq w'$.
    \end{itemize}
\end{definition}

\begin{definition}[Environment]
    Let $\Kmod$ be a Kripke structure. 
    A \demph{$w$-environment} is a function $\rho:\Var \to D_w$. 
\end{definition}
    Note that a $w$-environment is also a $w'$-environment for any $w' \geq w$. 
\begin{definition}[Satisfaction]
    For modal formula $A$, structure $\Kmod$ and $w$-environment $\rho$, we inductively define the judgement $\Kmod, w \models^\rho \sttr x A$ as follows, where $\rho[x:=d]$ is the map that sends variable $x$ to $d$ and agrees with $\rho$ on all other variables:
\begin{itemize}
    \item $\Kmod,w \models^\rho x : p$ if $\rho(x) \in \Prop_w(p)$; 
    \item $\Kmod,w \notmodels^\rho \bot$;
    \item $\Kmod,w \models^\rho \sttr x A \wedge B$ if $\Kmod,w \models^\rho \sttr x A$ and $\Kmod,w \models^\rho \sttr x B$;
    \item $\Kmod,w \models^\rho \sttr x A \vee B$ if $\Kmod,w \models^\rho \sttr x A$ or $\Kmod,w \models^\rho \sttr x B$;
    \item $\Kmod,w \models^\rho \sttr x A \limp B$ if for all $w' \geq w$, if $\Kmod,w' \models^\rho \sttr x A$ then  $\Kmod,w' \models^\rho \sttr x B$;
    \item $\Kmod,w \models^\rho x : \Box A$ if for all  $w' \geq w$ and $d \in D_{w'}$, if $\rho(x)R_{w'}d$, then $\Kmod,w' \models^{\rho[y:=d]} \sttr y A$;
    \item $\Kmod,w \models^\rho x : \Diamond A$ if there exists $d \in D_w$ such that $\rho(x)R_{w}d$ and $\Kmod,w \models^{\rho[y:=d]} \sttr y A$.
\end{itemize}
We write $\Kmod \models \sttr x A$ if $\Kmod, w \models^\rho \sttr x A$ for all worlds $w$ and $w$-environments $\rho$.
\end{definition}

The monotonicity lemma also holds in Kripke structures.

\begin{lemma}[Monotonicity lemma]
\label{lem:monotonicityPred}
    Let $\Kmod$ be a Kripke structure.
    If $w \leq w'$ and $\Kmod, w \models^\rho x : A$, then $\Kmod, w' \models^\rho x : A$.
\end{lemma}

To capture transitivity in Kripke structures, it is sufficient to require each $R_w$ to be transitive, which can be considered as a local condition on Kripke structures. However interpreting termination requires us to consider the interactions between $\leq$ and $R_w$ similarly to
the previous section. 

Let $\Kmod$ be a Kripke structure. 
    We write $D_W$ for the set of ordered pairs of the form $(w,d)$ with $w \in W$ and $d \in D_w$. We define the two binary relations $\leq_{D_W}, R_{D_W} \subseteq D_W \times D_W$ as:
    \begin{itemize}
        \item $(w,d) \leq_{D_W} (w',d')$ iff $w \leq w'$ and $d = d'$;
        \item $(w,d) R_{D_W} (w',d')$ iff $w = w'$ and $d R_{w} d'$.
    \end{itemize}

\begin{definition}
Write $\predIGL$ for the class of Kripke structures $\Kmod$
satisfying the following:
\begin{itemize}
    \item for all $w \in W$, $R_w$ is \emph{transitive}; and
    \item relation $(\leq_{D_W} ; R_{D_W})$ is \emph{terminating}, i.e., there are no infinite paths $(w_1, d_1) \leq_{D_W} (w_2,d_1) R_{D_W} (w_2, d_2) \leq_{D_W} (w_3, d_2) R_{D_W} (w_3,d_3) \dots$.
\end{itemize}
We write $\predIGL \models A$ to mean that $\Kmod \models \sttr x A$ for all $\Kmod \in \predIGL$ and any $x \in \Lab$.
\end{definition}

The same construction converting a
predicate structure into a 
birelational model
from \cite[Section 8.1.1]{Simpson94PhD} can be used to prove the following result. 

\begin{proposition}
\label{prop:birel-sat-implies-pred-sat}
If $\birelIGL \models A$ then $\predIGL \models A$.
\end{proposition}

\section{Completeness of a multi-succedent calculus via determinacy}
\label{sec:completeness}

Towards completeness we perform a countermodel construction using a proof search strategy based on an intuitionistic \emph{multi-succedent} calculus. 
This is inspired by analogous arguments for intuitionistic predicate logic, e.g.\ in \cite{Takeuti75}, but adapted to a non-wellfounded setting.
\begin{definition}
    [Multi-succedent intuitionistic calculus]
    The system $\multiintlabKfour$ is the restriction of $\labKfour$ where the $\Box$-right and $\limp$-right rules must have exactly one formula on the RHS.
    We write $\multiintlabGL$ for the class of $\infty$-proofs of $\multiintlabKfour$.
\end{definition}

The remainder of this section is devoted to proving the following completeness result:
\begin{theorem}
\label{thm:completeness-of-mlIGL}
    If $\predIGL \models A$ then $\multiintlabGL \proves x : A$.
\end{theorem}

Here we informally describe 
the construction of the proof search tree. For a more formal treatment of parts below we also refer to Appendix~\ref{app-countermodel}. During bottom-up proof search we will always proceed according to the three following phases in order of priority:
applications of rule $\trans$, applications of invertible rules (other than $\trans$), and application of non-invertible rules. The first two together we call the \demph{invertible phase}, the other the \demph{non-invertible phase}. 
In fact, invertibility of all rules except $\rightrule\Box$ and $\rightrule\limp$ is guaranteed by applying suitable contractions at the same time.
For instance, we apply the following derivable `macro' rules for 
$\limp $ on the left and $\Diamond$ on the right:
\[\small
\vliinf{\leftrule \limp}{}{\relats R, \Gamma , x:A\limp B \seqar \Delta}{\relats R, \Gamma, x:A\limp B \seqar \Delta, x:A  }{\relats R, \Gamma, x : A \limp B, x:B \seqar \Delta} 
\quad
\vlinf{\rightrule \Diamond}{}{\relats R, xRy,  \Gamma \seqar \Delta, x:\Diamond A}{\relats R, xRy, \Gamma \seqar \Delta, x:\Diamond A, y:A}
\]

By furthermore building weakening into the identity, i.e.\ allowing initial sequents of form $\relats R, \Gamma, x:p \seqar \Delta, x:p$, structural rules become redundant for proof search. In the invertible phase we can apply the rules in any order. The non-invertible phase creates predecessor nodes for each possible rule instance of $\rightrule \limp$ and $\rightrule \Box$.

In order to carry out our countermodel construction to show completeness of $\mlIGL$, we rely on two features of the proof search space that digress from usual countermodel constructions in intuitionistic predicate logic~\cite{Takeuti75}. The first important feature of this strategy is that the invertible phase is always finite and ends in so-called \emph{saturated sequents}, i.e., sequents for which any bottom-up rule application (other than $\id$ and $\leftrule \bot$) yields a premiss that is the same sequent, up to multiplicities (see Appendix~\ref{app-countermodel} for formal definition). Looking ahead to the countermodel, worlds $w$ are defined on the basis of invertible phases and will as a result all have a finite domain $D_w$.

\begin{lemma}
\label{lem:invertiblefinite}
    Following the proof search strategy described above, each invertible phase constructs a finite subtree that has saturated sequents at its leaves.
\end{lemma}

Secondly, and perhaps more importantly, we employ a technique from non-wellfounded proof theory to help us organise the countermodel constructed from a failed proof search: we appeal to determinacy of a \emph{proof search game}. To understand the motivation here, a classical countermodel-from-failed-proof search argument proceeds (very roughly) as follows: 
(1) assume a formula is not provable; 
(2) for each rule instance there must be an unprovable premiss; 
(3) continue in this way to construct an (infinite) `unprovable' branch; 
(4) extract a countermodel from this branch. 
In our setting we will need the branch obtained through the process above to be  \emph{not progressing} in order to deduce that the structure we extract is indeed one of $\predIGL$. 
However the local nature of the process above does not at all guarantee that this will be the case. For this, we rely on \emph{(lightface) analytic determinacy}, which is equivalent to the existence of $0^\sharp$ over $\mathrm{ZFC}$ \cite{harrington_1978}, of the corresponding proof search game. 
As we only use it as a tool and it is not the main focus of our work, we refer to \cref{app-countermodel} for more details.

\begin{proposition}
\label{prop:denier-subtree}
    Given an unprovable sequent $S$ there is a subtree $T$ of the proof search space rooted at $S$, closed under bottom-up non-invertible rule application\footnote{I.e.\ if $S_0\in T$ concludes some non-invertible step with premiss $S_1$, then also $S_1\in T$.} such that each infinite branch of $T$ is not progressing.
\end{proposition}

The properties in \cref{lem:invertiblefinite} and \cref{prop:denier-subtree} enable us to construct a countermodel: 

\begin{theorem}[Countermodel construction]
\label{thm:countermodel-construction}
    If $\multiintlabGL \not \proves \relats R,\Gamma \seqar \Delta$, then there is a structure $\Kmod \in \predIGL$ with $w$-environment $\rho$ such that for labelled formulas $x : A$ we have
    \begin{itemize}
        \item $\text{if }x : A \in \Gamma, \text{ then } \Kmod, w \models^\rho \STsub{x}{A}, \text{ and, }$
        \item $\text{if } x : A \in \Delta, \text{ then } \Kmod, w \notmodels^\rho \STsub{x}{A}.$
    \end{itemize}
\end{theorem}

From here \cref{thm:completeness-of-mlIGL} easily follows.
Notice that we did not use rule $\cut$ in the proof search strategy so we can actually conclude a stronger result: 
\begin{corollary}
    $\multiintlabGL$ is \emph{cut-free} complete over $\predIGL$.
\end{corollary}
\section{Completeness of $\labIGL$ via (partial) cut-elimination}
\label{sec:cut-elim}

To obtain completeness of $\labIGL$ for $\birelIGL$, we will simulate $\mlIGL$ using cuts in an extension of $\labIGL$ that allows reasoning over disjunctions of labelled formulas. 
We apply a (partial) cut-elimination procedure to eliminate these disjunctions, and then note that any resulting proof is already one of $\labIGL$.

\medskip

\noindent
We shall use metavariables $\phi,\psi$ etc.\ to vary over \emph{disjunctions} of labelled formulas. I.e.\
\(
\phi,\psi, \dots \quad ::= \quad (x:A) \quad | \quad \phi \lor \psi
\).

\begin{definition}
    The system $\dlabIKfour $ is the extension of $\labIKfour$ by duly adapting identity, cut, structural and $\lor $ rules to allow for $\phi$-formulas. In particular it has the following $\lor $ rules:
    \[
    \small
    \vliinf{\leftrule \lor}{}{\relats R, \Gamma, \phi_0 \lor \phi_1 \seqar \psi}{\relats R, \Gamma, \phi_0 \seqar \psi}{\relats R, \Gamma , \phi_1 \seqar \psi}
    \qquad
    \vlinf{\rightrule \lor}{i \in \{0,1\}}{\relats R, \Gamma \seqar \phi_0\lor \phi_1}{\relats R, \Gamma \seqar \phi_i}
    \]
\end{definition}

The \demph{degree} of a formula $(x_1: A_1) \lor \cdots \lor (x_d:A_d)$ is $d$. 
The degree of a cut is the degree $d$ of its cut-formula, in which case we say it is a $d$-cut.
The degree of a preproof is the maximum degree of its cuts (when this is well-defined). 

\begin{proposition}
\label{prop:cut-free-multi-to-dis-lab-cut-sys}
If $\mlIKfour$ has a cut-free $\infty$-proof of $x:A$, $\dlabIKfour$ has~one~of~bounded~degree. 
\end{proposition}

Moreover, immediately from definitions, we have:
\begin{observation}
\label{prop:dis-lab-sys-only-lab}
    A $\dlabIKfour$ $\infty$-proof containing only labelled formulas is a $\labIKfour$ $\infty$-proof.
\end{observation}

Thus, to conclude completeness of $\labIGL$ for $\birelIGL$, it suffices to eliminate the use of disjunctions of labelled formulas in $\dlabIKfour$ $\infty$-proofs.
We will prove this by a partial cut-elimination procedure, reducing cuts over disjunctions of labelled formulas until they are on labelled formulas.

For the remainder of this section we work only with preproofs without thinning $\thinning$, by \cref{obs:thinning}.
Recall that this means that relational contexts are \emph{growing}, bottom-up.
For a sequent $S$ write $\relcxt S$ for its relational context. I.e.\ if $S$ is $\relats R, \Gamma \seqar \phi$, then $\relcxt S := \relats R$.

\begin{lemma}
    [Invertibility]
    \label{lem:invertibility-of-orleft}
    If $\dlabIKfour $ has a $\infty$-proof $P$ of $ \relats R, \Gamma , \phi_0 \lor \phi_1 \seqar \psi$ then it also has $\infty$-proofs $P_i$ of $ \relats R, \Gamma, \phi_i \seqar \psi$, for $i\in \{0,1\}$.
    Moreover, for each branch 
    $
    (S_i)_{i<\omega}
    $ of $P_i$ there is a branch 
    $
    (S_i')_{i<\omega}
    $ of $P$ such that $\relcxt{S_i} \subseteq \relcxt {S_i'}$ for all $i<\omega$.
\end{lemma}

From here the key cut-reduction for $\phi$-formulas is simply,
\[
\small
\vlderivationnc{
\vliin{\cut}{}{\relats R, \Gamma, \Gamma' \seqar \phi}{
    \vlin{\rightrule \lor}{}{\relats R, \Gamma \seqar \chi_0 \lor \chi_1}{
    \vltr{P}{\relats R, \Gamma, \seqar \chi_i }{\vlhy{\ }}{\vlhy{\ }}{\vlhy{\ }}
    }
}{
    \vltr{Q}{\relats R, \Gamma', \chi_0 \lor \chi_1 \seqar \phi}{\vlhy{\ }}{\vlhy{\ }}{\vlhy{\ }}
}
}
\qquad \leadsto \qquad 
\vlderivationnc{
\vliin{\cut}{}{\relats R, \Gamma, \Gamma' \seqar \phi}{
    \vltr{P}{\relats R, \Gamma \seqar \chi_i }{\vlhy{\ }}{\vlhy{\ }}{\vlhy{\ }}
}{
    \vltr{Q_i}{\relats R, \Gamma' , \chi_i \seqar \phi}{\vlhy{\ }}{\vlhy{\ }}{\vlhy{\ }}
}
}
\]
where $Q_i$ is obtained by \cref{lem:invertibility-of-orleft} above.
Note that this reduction `produces' a cut of lower complexity.
Commutative cut-reduction cases, where the cut-formula is not principal on the left, are standard and always produce.
Note that, thanks to invertibility, we do not consider commutations over the right premiss of a $\phi$-cut. 
However commutative cases may \emph{increase} the heights of progress points; for instance when commuting over a $\leftrule \Box$ step,
\[
\small
\vlderivationnc{
\vliin{\cut}{}{\relats R, xRy, \Gamma, \Gamma', x:\Box A \seqar \phi}{
    \vlin{\leftrule \Box}{}{\relats R, xRy, \Gamma , x:\Box A \seqar \chi}{
    \vltr{P}{\relats R, xRy,\Gamma, y: A \seqar \chi }{\vlhy{\ }}{\vlhy{\ }}{\vlhy{\ }}
    }
}{
    \vltr{Q}{\relats R, xRy, \Gamma' , \chi\seqar \phi }{\vlhy{\ }}{\vlhy{\ }}{\vlhy{\ }}
}
}
\quad \leadsto\quad
\vlderivationnc{
\vlin{\leftrule \Box}{}{\relats R, xRy, \Gamma, \Gamma' , x:\Box A \seqar \phi}{
\vliin{\cut}{}{\relats R, xRy, \Gamma, \Gamma', y:A \seqar \phi}{
    \vltr{P}{\relats R, xRy,\Gamma, y: A \seqar \chi }{\vlhy{\ }}{\vlhy{\ }}{\vlhy{\ }}
}{
    \vltr{Q}{\relats R, xRy, \Gamma' , \chi\seqar \phi }{\vlhy{\ }}{\vlhy{\ }}{\vlhy{\ }}
}
}
}
\]
observe that progress points in $Q$ have been raised.
Thus, to show that the limit of cut-reduction is progressing we will need appropriate invariants, requiring additional notions.

\medskip

A \demph{bar} of a preproof is a (necessarily finite, by K\"onig's Lemma) antichain intersecting each infinite branch. 
Each bar $B$ induces a (necessarily finite) subtree $\downset B $ of nodes beneath (and including) it.
Given a cut-reduction $ P \leadsto_\cutred P'$ we associate to each bar $B$ of $P$ a bar $\mathsf r(B)$ of $P'$ in the natural way. 
In particular we have that $\mathsf r(B)$ satisfies the following properties:

\begin{lemma}
[Trace preservation]
\label{lem:rel-cxt-from-bars-preserved}
   If $P \leadsto_\cutred P'$ and $B$ a bar of $P$, for any sequent $S' \in \mathsf r(B)$ there is a sequent $S\in B$ such that $\relcxt {S'} \supseteq \relcxt S$. 
\end{lemma}
This lemma allows us to keep track of a fixed amount of progress information during cut-elimination, sidestepping the issue that commutative cases raise progress points.
Note that we really need the $\relcxt {S'} \supseteq \relcxt S$ due to the $\leftrule \Diamond$ commutative case:
\[
\toks0={.6}
\small
\vlderivationnc{
\vliin{\cut}{}{\relats R, \Gamma, \Gamma', x:\Diamond A \seqar \phi}{
    \vlin{\leftrule \Diamond}{}{\relats R, \Gamma, x: \Diamond A \seqar \chi}{
    \vltr{P}{\relats R, xRy, \Gamma , y:A \seqar \chi}{\vlhy{\ }}{\vlhy{\ }}{\vlhy{\ }}
    }
}{
    \vltr Q {\relats R, \Gamma' , \chi \seqar \phi}{\vlhy{\ }}{\vlhy{\ }}{\vlhy{\ }}
}
}
\quad \leadsto \quad
\vlderivationnc{
\vlin{\leftrule \Diamond}{}{\relats R, \Gamma, \Gamma' , x:\Diamond A \seqar \phi}{
\vliin{\cut}{}{\relats R, xRy, \Gamma, \Gamma', y:A \seqar \phi}{
    \vltr{P}{\relats R, xRy, \Gamma , y:A \seqar \chi}{\vlhy{\ }}{\vlhy{\ }}{\vlhy{\ }}
}{
    \vltrf {xRy,Q} {xRy,\relats R, \Gamma' , \chi \seqar \phi}{\vlhy{\quad }}{\vlhy{\quad }}{\vlhy{\quad }}{\the\toks0}
}
}
}
\]
where the preproof $xRy,Q$ is obtained from $Q$ by prepending $xRy$ to the LHS of each sequent.
From here we can effectively reduce infinitary cut-elimination to finitary cut-elimination, by eliminating cuts beneath higher and higher bars. 
The step case is given by:

\begin{lemma}
[Productivity]
\label{lem:push-cuts-above-bar}
For any $\infty$-proof $P$ and bar $B$ there is a sequence of cut-reductions $P=P_0 \leadsto_{\cutred_1} \cdots \leadsto_{\cutred_n} P_n$ such that $\cutred_n \cdots \cutred_1 (B)$ has no $d$-cuts beneath it in $P_n$.
\end{lemma}

The argument proceeds in a relatively standard way:
    by induction on the number of $d$-cuts in $\downset B$, with a subinduction on the multiset of the distances of topmost $d$-cuts in $\downset B$ from $B$, where the distance is the length of the shortest path from the cut to $B$. 
Applying \cref{lem:push-cuts-above-bar} to higher and higher bars allows us to reduce cut-degrees as required:

\begin{proposition}
 [Degree-reduction]
 \label{prop:cut-degree-reduction}
    For each $\infty$-proof of $\dlabIKfour$ of $x:A$ of degree $d$, there is one of degree $<d$.
\end{proposition}
Importantly here, the preservation of progress in the limit crucially relies on \cref{lem:rel-cxt-from-bars-preserved}, under K\"onig's Lemma.
Finally by induction on degree we have:

\begin{corollary}
[Partial cut-elimination]
\label{cor:partial-cut-elim}
    If $\dlabIKfour $ has a bounded-degree $\infty$-proof of $x:A$, then it has one containing only labelled formulas.
\end{corollary}

From here we have our desired converse to \cref{thm:soundness}, following from \cref{prop:dis-lab-sys-only-lab,prop:birel-sat-implies-pred-sat,prop:cut-free-multi-to-dis-lab-cut-sys},
\cref{thm:completeness-of-mlIGL,cor:partial-cut-elim}:

\begin{theorem}
[Completeness]\label{thm:completeness}
    If $\birelIGL\models A$ then $\labIGL \proves x : A$.
\end{theorem}

\section{Conclusions}\label{sec:concs}

We have
recovered several intuitionistic formulations of $\GL$, both syntactically and semantically, in the tradition of Simpson, on intuitionistic modal logic \cite{Simpson94PhD}, and Simpson, on non-wellfounded proofs \cite{Simpson17}.
We proved the equivalence of all these formulations, cf.~\cref{fig:tour}, motivating the following definition:

\begin{definition}
    $\IGL$ is the modal logic given by any/all of the nodes of \cref{fig:tour}.
\end{definition}

\noindent
Thanks to the methodology we followed, $\IGL$ satisfies Simpson's \emph{requirements} from \cite{Simpson94PhD}. 
$\IGL$  
also interprets (classical) $\GL$ along the G{\"o}del-Gentzen negative translation.

It would be interesting to examine $\IGL$ as a logic of provability, returning to the origins of~$\GL$.
In particular $\IGL$ (even $\IK$) has a \emph{normal} $\Diamond$, distributing over $\lor$, so let us point out that it is not sound to interpret $\Diamond$ as consistency $\lnot \Box \lnot$.
At the same time (effective) model-theoretic readings of the $\Diamond$ stumble on the consequence (already of $\IK$) $\Box(A\limp B) \limp \Diamond A \limp \Diamond B$.
We expect $\IGL$ to rather correspond to the provability logic of some \emph{model} of $\HA$.

To this end it would be pertinent to develop a bona fide \emph{axiomatisation} of $\IGL$. 
As far as we know it is possible that $\IGL$ is simply the extension of $\IKfour $ by L\"ob's axiom but, similarly to other `$\mathsf I$' logics, we expect that further axioms involving $\Diamond$ will be necessary. 
On a related note, we believe that a full cut-elimination result holds for $\labIGL$, in particular by extending our argument for $\dlabIGL$.
Crucial here is that no cut-reductions \emph{delete} progress points, since we do not cut on relational atoms (cf.~a key $\Box$ reduction).
A full development of this is beyond the scope (and allocated space) of this work.

\nocite*

\bibliographystyle{splncs04}
\providecommand{\noopsort}[1]{}


\appendix
\section{Appendix for \cref{sec:gl}}
\label{app:gl}

\begin{proof}
[Proof of \cref{obs:thinning}]
Delete every thinning step and replace each relational context $\relats R$ with the union of all relational contexts beneath it.
\end{proof}

\begin{proof}
[Proof of \cref{prop:soundness-of-labGL}]
First, given a model $\M = (W,\interp \M R,V)$,
an \demph{($\M$-)assignment} for a sequent $S = \relats R, \Gamma \seqar \Delta$ is a function $\rho : \fv S \to W$ such that whenever $xRy$ occurs in $\relats R$ we have $\rho(x)R^\M \rho(y)$.
We write $\Mmod, \rho \models S$ if:
$\Mmod, \rho(x) \models A$ for some $x:A$ in $\Delta$, whenever $\Mmod, \rho(y) \notmodels B$ for all $y:B$ in $\Gamma$.

Suppose otherwise. Let $\pi $ be a $\infty$-proof of $x:A$ in $\labKfour$, and let $\M = (W,\interp \M R,V)$ be a structure in which $\interp \M R$ is transitive yet $\M, w \notmodels A$. We shall show that $\interp \M R$ is not terminating on $W$.

Set $\rho_0(x) = w$. 
By local soundness of the rules, construct an infinite branch $(S_i)_{i<\omega} $ and assignments $\rho_i:\fv {S_i}  \to W$ such that $\M, \rho_i \notmodels  S_i$ always satisfying:
$\rho_{i}(z) = \rho_{i+1}(z)$, whenever $z \in \fv{S_i}\cap \fv{S_{i+1}}$.

Now, let $(x_i)_{i<\omega}$ be a progressing trace along $(S_i)_{i<\omega}$.
By inspection of the rules we have that, at any progress point $x_i$, necessarily $\rho_i(x_i)R^\Mmod\rho_{i+1}(x_{i+1})$, otherwise $x_i = x_{i+1}$ and so $\rho_i(x_i) = \rho_{i+1}(x_{i+1})$ by above.
Thus $(x_i)_{i<\omega}$ induces an infinite path along $\interp \M R$, and so $\interp \M R$ is not terminating.
\end{proof}

\begin{proof}
    [Proof of \cref{prop:completeness-labgl}]
    It is well-known that $\labKfour$ is complete for $\Kfour$, and so derives all its axioms~\cite{Negri05}. 
    Together with the derivation of L\"ob's axiom in \cref{ex:loeb}, completeness follows by closure of $\labGL$ under necessitation (simply by $\rightrule \Box$) and modus ponens (by $\cut$). 
\end{proof}
\section{Appendix for \cref{sec:soundness}}
\label{app-soundness}

\begin{proof}[Proof of Observation \ref{lem:models-form-seq}]
    For left to right, suppose~$v \geq w$ such that~$\Bmod, v \models \bigwedge_{i} B_i$. Then for~$\calI(x)=v$ we have $\Bmod, \calI(x) \models B_i$ for all~$i \leq n$. So by assumption,~$\Bmod, v \models A$ as desired. For the other direction, let~$\calI$ be an interpretation, for which its domain only consists of~$x$. Let~$\Bmod, \calI(x) \models B_i$ for all~$i \leq n$. Then,~$\Bmod, \calI(x) \models \bigwedge_{i} B_i$ and so by assumption~$\Bmod, \calI(x) \models A$.
\end{proof}

\begin{proof}[Proof of Lemma \ref{lem:liftinglemma}]
    In our proof we rely on the lifting lemma in \cite{Simpson94PhD} (Lemma~8.1.3), where $S$ is assumed to be a tree-like sequent. We rely on the fact that each quasi-tree $\relats R$ is supported by a unique tree~$\relats R'$, i.e., $\relats R'$ is a tree and $\Lab(\relats R')=\Lab(\relats R)$. Consider sequent $S = (\relats R' , \Gamma \seqar x : A)$. Note that $\calI$ is also an interpretation of $S'$, so we can now apply Lemma 8.1.3 from \cite{Simpson94PhD} to obtain an interpretation $\calI'$ of $S'$ such that $\calI'(x) = w$ and for all~$y \in \Lab(S)$ we have~$\calI'(y) \geq \calI(y)$. Since $\Bmod$ is transitive, $\calI'$ is also an interpretation for $S$, which concludes the proof.
\end{proof}

\begin{proof}[Proof of Proposition \ref{prop:localsoundness}]
    Let us show the proofs for the rules~$\rightrule\limp, \leftrule\limp, \leftrule\Box, \rightrule\Box, \leftrule\Diamond$, $\rightrule\Diamond$. $\leftrule \wk$, and $\trans$. For $\rightrule \limp$, $\rightrule \Box$ we use the lifting lemma (Lemma \ref{lem:liftinglemma}).

    For rule $\rightrule \limp$, by assumption, $\Bmod, \calI \notmodels \relats R, \Gamma \seqar x : A \limp B$. In other words,  $\Bmod, \calI(y) \models D$ for all~$y : D \in \Gamma$ and $\Bmod, \calI(x) \notmodels A \limp B$. So there exists~$w \geq \calI(x)$ such that $\Bmod, w \models A$ and $\Bmod, w \notmodels B$. By the lifting lemma (Lemma~\ref{lem:liftinglemma}), there is an interpretation~$\calI'$ of~$S$ such that~$\calI'(x)=w$ and~$\calI'(z) \geq \calI(z)$ for all~$z \in \relats \Lab(S)$. By monotonicity (Lemma~\ref{lem:monotonicity}) we have that~$\Bmod, \calI'(y) \models D$ for all~$y : D \in \Gamma$. Also~$\Bmod, \calI'(x) \models A$ and $\Bmod, \calI'(x) \notmodels B$. Therefore,~$\Bmod, \calI' \notmodels \relats R, \Gamma, x : A \seqar x : B$ as desired.

    For rule $\leftrule\limp$, by assumption $\Bmod, \calI \notmodels \relats R, \Gamma_1, \Gamma_2, x : A \limp B \seqar z : C$. So, $\Bmod, \calI(y) \models D$ for all~$y : D \in \Gamma_1 \cup \Gamma_2$,~$\Bmod,\calI(x) \models A \limp B$, and $\Bmod, \calI(z) \notmodels C$. Either $\Bmod, \calI(x) \models A$ or $\Bmod, \calI(x) \notmodels A$. In the latter case, $\Bmod, \calI \notmodels \relats R, \Gamma_1 \seqar x : A$. In the former case, since~$\Bmod,\calI(x) \models A \limp B$, also $\Bmod, \calI(x) \models x: B$, and hence $\Bmod, \calI \notmodels \relats R, \Gamma_2, x : B \seqar z : C$.

    Consider rule $\leftrule \Box$. By assumption we have that $\Bmod, \calI \notmodels \relats R, xRy, \Gamma, x : \Box A  \seqar z : C$. As~$\calI$ is an interpretation we have~$\calI(x)R^\Bmod \calI(y)$. From~$\Bmod,\calI(x) \models \Box A$ it follows that~$\Bmod, \calI(y) \models A$. Hence $\Bmod, \calI \notmodels \relats R, xRy, \Gamma, y : A \seqar z : C$.
    
    For rule $\rightrule \Box$, we have $\Bmod, \calI \notmodels \relats R, \Gamma \seqar x : \Box A$. So there exist worlds $w$ and $v$ such that $\calI(x)\leq wR^\Bmod v$ and $\Bmod, v \notmodels A$. By the lifting lemma we know that there is an interpretation~$\calI'$ of~$S$ satisfying $\calI'(x)=w$ and~$\calI'(z) \geq \calI(z)$ for all~$z \in \Lab(S)$. Let us define an interpretation~$\calJ$ of premiss $\relats R, xRy, \Gamma \seqar y : A$ with fresh label~$y$, that agrees with $\calI'$ on all labels from~$\relats R \cup \{x \}$, and $\calJ(y)=v$. This is a well-defined interpretation by the observation that~$\calJ(x) R^\Bmod \calJ(y)$. By monotonicity (Lemma~\ref{lem:monotonicity}), we have $\Bmod, \calJ' \notmodels \relats R, xRy, \Gamma \seqar y : A$.

    For rule $\leftrule \Diamond$, we have $\Bmod, \calI \notmodels \relats R, \Gamma, x : \Diamond A \seqar z : C$. So there exists a~$v$ such that $\calI(x)R^\Bmod v$ and~$\Bmod, v \models A$. We define a new interpretation~$\calJ$ of tree~$\relats R \cup \{xRy\}$ with fresh label~$y$, that agrees with $\calI$ on all labels from~$\relats R$, and $\calJ(y)=v$. We have~$\Bmod, \calJ \notmodels \relats R, xRy, \Gamma, y : A \seqar z : C$.

    For rule $\rightrule \Diamond$ we have $\Bmod, \calI \notmodels \relats R,xRy, \Gamma \seqar x : \Diamond A$. Since,~$\calI(x)R^\Bmod \calI(y)$ and $\Bmod, \calI(x) \notmodels \Diamond A$, we have $\Bmod, \calI \notmodels \relats R,xRy, \Gamma \seqar y : A$.

    For rule $\leftrule \wk$ we have $\Bmod, \calI \notmodels \relats R, \Gamma, x  :A \seqar z : C$. Since sequent $\relats R, \Gamma, x  :A \seqar z : C$ is assumed to be quasi-tree-like, $\relats R, \Gamma \seqar z : C$ will contain the same labels. So $\calI$ will be an interpretation for $\relats R, \Gamma \seqar \Delta$ and $\Bmod, \calI \notmodels \relats R, \Gamma \seqar z : C$.

    Finally, let us consider rule $\trans$. We have $\Bmod, \calI \notmodels \relats R, xRy,yRz, \Gamma \seqar z : C$. It follows immediately from transitivity of $\calB$ that $\Bmod, \calI \notmodels \relats R, xRy,yRz, xRz, \Gamma \seqar z : C$.
\end{proof}

\begin{proof}
[Proof of \cref{thm:soundness}]
    We proceed by contradiction and we use the fact that rule $\thinning$ is admissible in $\labIGL$. Suppose that~$\labIGL \proves \relats R , \Gamma \seqar x : A$, but suppose that~$\Bmod, \calI \notmodels \relats R , \Gamma \seqar x : A$ for some model~$\Bmod \in \birelIGL$ with interpretation~$\calI$ of~$\relats R$ into $\Bmod$. Let~$\pi$ be the $\infty$-proof of~$\relats R , \Gamma \seqar x : A$ in~$\labIGL$ without the use of the thinning rule~$\thinning$. By the local soundness of the rules (\cref{prop:localsoundness}), we can construct a (possibly) infinite sequence of labelled sequents~$(\relats R_i , \Gamma_i \seqar x_i : A_i)_i$ and interpretations~$\calI_i$ of~$\relats R_i$ into $\Bmod$ such that
    \begin{enumerate}
        \item 
        $(\relats R_0 , \Gamma_0 \seqar x_0 : A_0) = (\relats R , \Gamma \seqar x : A)$ and $\calI_0 = \calI$;
        \item
        $\relats R_{i+1} , \Gamma_{i+1} \seqar x_{i+1} : A_{i+1}$ is a premiss of the rule in~$\pi$ with conclusion $\relats R_i , \Gamma_i \seqar x_i : A_i$;
        \item
        $\Bmod, \calI_i \notmodels \relats R_i ; \Gamma_i \seqar x_i : A_i$; and
        \item
        $\calI_{i+1}(z) \geq \calI_i(z)$ for all $z \in \dom(\calI_i)$. 
    \end{enumerate}
    If the sequence is finite, we end in a sequent $\relats R_n ; \Gamma_n \seqar x_n : A_n$ derived by~$\id$ or~$\leftrule \bot$ which is true in every model by definition, deriving a contradiction with~$\Bmod, \calI_n \notmodels \relats R_n ; \Gamma_n \seqar x_n : A_n$. If the sequence is infinite, it corresponds to an infinite branch in~$\pi$ which has a progressing trace $(y_i)_{i \geq k}$ by definition. So,
    \begin{itemize}
        \item 
        $y_i = y_{i+1}$ in which case~$\calI_i(y_i) \leq \calI_{i+1}(y_{i+1})$ by construction, or;
        \item
        $y_iRy_{i+1}$ appears in~$\relats R_i$, which means that $\calI_i(y_i) R^\Bmod \calI_i(y_{i+1}) \leq \calI_{i+1}(y_{i+1})$.
    \end{itemize}
    The second property occurs infinitely often, which creates a sequence of worlds $(\calI_i(y_i))_{i \geq k}$ in model $\Bmod$ which is increasing in the order~$(\leq ; R^\Bmod)$, which contradicts the fact that~$(\leq ; R^\Bmod)$ is terminating.
\end{proof}
\section{Appendix for \cref{sec:predmodels}}
\label{app-predmodels}

\begin{proof}[Proof of Proposition \ref{prop:birel-sat-implies-pred-sat}]
    Let $\Kmod \in \predIGL$. Following the construction from Section 8.1.1 in~\cite{Simpson94PhD}, we convert $\Kmod$ into a birelational model $\Bmod_\Kmod = (W', \leq',R',V')$ as follow where unary predicate symbol $p \in \Prop$ is considered to be a propositional variable:
   \begin{alignat*}{3}
    W' && \ \ = &\ \ D_W;\\
    \leq' && \ \ = & \ \ \leq_{D_W};\\
    R'  && \ \ = & \ \ R_{D_W};\\
    V'((w,d)) && \ \ = & \ \ \{ p \mid  d \in \Prop_w(p)\}.
\end{alignat*}
Indeed, $\leq'$ is a partial order and $V'$ is monotone. Also the frame properties (F1) and (F2) are easy to check \cite{Simpson94PhD}. Model $\Bmod_\Kmod$ is in $\birelIGL$, because $R'$ is transitive as each $R_w$ is transitive by definition of $\Kmod$ and therefore $R_{D_W}$ is transitive. Moreover, $(\leq' ; R')$ is terminating simply because $(\leq_{D_W};R_{D_w})$ is required to be terminating.

Let $w \in W$ and let $\rho$ be a $w$-environment. By a straightforward induction on $A$ we can prove that $\Kmod, w \models^{\rho[x:=d]} \sttr x A$ if and only if $\Bmod_\Kmod, (w,d) \models A$. Therefore, since $\birelIGL \models A$ we have that $\Bmod_\Kmod, (w,d) \models A$ for all $(w,d) \in D_W$ and so $\Kmod, w \models^{\rho[x:=d]} x : A$ for all $(w,d) \in D_W$. Hence $\Kmod, w \models^\rho \sttr x A$.
\end{proof}
\section{Appendix for \cref{sec:completeness}}
\label{app-countermodel}

We start this appendix by setting up the proof search strategy in a more formal manner than explained in the main text. 

\begin{definition}
\label{def:invertible-rule}
    A rule of the form
    \[
    \vlinf{}{}{S}{S_1 \quad \dots \quad S_n}
    \]
    is \demph{invertible}, if for every instance of the rule, for all $i$, the instance of $S_i$ is provable whenever the instance of $S$ is provable.
\end{definition}

One can show by coinduction that all rules with built-in contraction and built-in weakening, except $\rightrule \limp$ and $\rightrule \Box$, are invertible.

In order to structure the proof search strategy we define saturated sequents.

\begin{definition}[Saturation]
    Let $\relats R, \Gamma \seqar \Delta$ be a sequent that contains labelled formula $x:A$. We say that $x : A$ is \demph{saturated} in $\relats R, \Gamma \seqar \Delta$ if the following conditions hold based on the form of $A$:  
    \begin{itemize}
        \item $A$ equals $\bot$ or $p$ for some $p \in \Prop$;
        \item if $A = A_1 \wedge A_2$ and $x:A \in \Gamma$, then $x:A_1 \in \Gamma$ and $x:A_2 \in \Gamma$;
        \item if $A = A_1 \wedge A_2$ and $x:A \in \Delta$, then $x:A_1 \in \Delta$ or $x:A_2 \in \Delta$;
        \item if $A = A_1 \vee A_2$ and $x:A \in \Gamma$, then $x:A_1 \in \Gamma$ or $x:A_2 \in \Gamma$;
        \item if $A = A_1 \vee A_2$ and $x:A \in \Delta$, then $x:A_1 \in \Delta$ and $x:A_2 \in \Delta$;
        \item if $A = A_1 \limp A_2$ and $x:A \in \Gamma$, then $x:A_1 \in \Delta$ or $x:A_2 \in \Gamma$;
        \item if $A = \Box A_1$ and $x:A \in \Gamma$, then for all $y$ with $xRy \in \relats R$, $y : A_1 \in \Gamma$;
        \item if $A = \Diamond A_1$ and $x:A \in \Gamma$, there exists $y$ s.t.~$xRy \in \relats R$ and $y : A_1 \in \Gamma$;
        \item if $A = \Diamond A_1$ and $x:A \in \Delta$, then for all $y$ with $xRy \in \relats R$, $y : A_1 \in \Delta$.
    \end{itemize}
    A sequent $\relats R; \Gamma \seqar \Delta$ is \demph{saturated} if all its formulas are saturated and $\relats R$ is transitively closed.
\end{definition}

Now we define the proof search strategy. During bottom-up proof search we will always proceed according to the three following phases in order of priority: applications of rule $\trans$, applications of invertible rules other than $\trans$, and application of non-invertible rules. The first two together we call the invertible phase, the other the non-invertible phase. More precisely, when possible, we apply rule $\trans$. If not, if possible, we apply another invertible rule for which the principal formula is not saturated and create predecessor node(s) labelled by the premiss(s). We will argue below that each invertible phase is finite and ends in saturated sequents. When $\relats R, \Gamma \seqar \Delta$ is saturated, we create predecessor nodes for each possible rule instance of $\rightrule \limp$ and $\rightrule \Box$, i.e., for each $x : A \impl B \in \Delta$ we create predecessor $\relats R, \Gamma, x : A \seqar x : B$, and for each $x : \Box A \in \Delta$ we create predecessor node $\relats R, xRy, \Gamma \seqar y : A$, where $y$ is fresh.

Now we prove Lemma~\ref{lem:invertiblefinite} stated in the main text, i.e., we show that given a proof search tree as defined above, each invertible phase is a finite subtree with saturated sequents in its leaves.

\begin{proof}[Proof of Lemma~\ref{lem:invertiblefinite}]
    By construction of the invertible phase, each leaf is saturated. Let $\relats R, \Gamma \seqar \Delta$ be the start sequent of the invertible phase. We argue that the invertible phase terminates. First observe that, bottom-up, the number of labelled formulas in a sequent increase for each rule because we only apply a rule in the invertible phase when the concerned principal formulas is not saturated. To ensure termination, we show that the number of labelled formulas of each sequent in the invertible phase is bounded. 
    
    Let $\relats R',\Gamma' \seqar \Delta'$ be a sequent in the invertible phase. We first show that the number of different labels in $\Gamma'$ is bounded. A new label in $\Gamma'$ can only be created via the rules $\leftrule \Box$ and~$\leftrule \Diamond$. Note that rule $\rightrule \limp$ can also create a new label in $\Gamma'$, but this rule is not part of the invertible phase. Therefore, each new formula in $\Gamma'$ is a subformula of another formula in $\Gamma'$. As there are only finitely many subformulas, and we only apply the rules to non-saturated sequents, one can only apply rules $\leftrule \Box$ and $\leftrule \Diamond$ finitely many times. Hence, there are only finitely many extra labels in $\Gamma'$.

    Now for each label $x$, again by saturation and the fact that new formulas in the sequent are subformulas of the sequent, one can only create finitely many formulas labelled by $x$. Since the number of labels is bounded, the number of formulas occurring in $\relats R',\Gamma' \seqar \Delta'$ is bounded.
\end{proof}

Now we set up the proof of Proposition~\ref{prop:denier-subtree} using techniques from non-wellfounded proof theory. As mentioned in the main text, a classical countermodel-from-failed-proof search argument proceeds (very roughly) as follows: (1) Assume a formula is not provable, (2) For each rule instance there must be an unprovable premiss, (3) Continue in this way to construct an (infinite) `unprovable' branch, (4) Extract a countermodel from this branch. In our setting we will need the branch obtained through the process above to be  \emph{not progressing} in order to deduce that the structure we extract is indeed one of $\predIGL$. However the local nature of the process above does not at all guarantee that this will be the case. 

In the setting of classical $\GL$, the invertibility of the rules of $\labGL$ allow us to nonetheless carry out a na\"ive such argument. For the intuitionistic setting, where invertibe and non-invertible rules may interact in a nontrivial way, we must use more powerful machinery.

\begin{definition}
    [Proof search game]
    \label{def:proofsearchgame}
    For a system $\system$, the \demph{proof search game} is a game between two players, Prover ($\prover$) and Denier ($\denier$), played as follows:
    \begin{itemize}
        \item The game initialises on a sequent $S$ of $\system$.
        \item $\prover$ chooses a rule instance of $\system$ that has $S$ as its conclusion.
        \item $\denier$ chooses a premiss of that rules instance, and they repeat the process.
    \end{itemize}
    A \demph{play} of the game is a maximal sequence of moves as described above.
    An infinite play of the game is \demph{won} by $\prover$ (and \demph{lost} by $\denier$) just if it has a progressing thread. {(If deadlock is reached, the player with no valid move loses.)}
\end{definition}

It is not hard to see that $\prover$-strategies from $S$ are \emph{just} $\system$-preproofs of $S$, and that \emph{winning} $\prover$-strategies from $S$ are just $\system$-proofs of $S$.
What is more difficult is to show that winning $\denier$-strategies induce countermodels; this is what we do in the next subsection.

Before that, we better state the following:

\begin{proposition}
    [Determinacy, $\exists 0^\#$]
    \label{prop:determinacy}
    The proof search game for $\system$ is determined.
\end{proposition}

This is a consequence of \emph{analytic determinacy} as the progressing condition is indeed (lightface) $\Sigma^1_1$: ``there \emph{exists} an infinite thread that progresses infinitely often'' \cite{harrington_1978}. 

One can view the proof search space of $\multiintlabGL$ as a collection of plays of the proof search game (Definition \ref{def:proofsearchgame}), where each branch represents a play. We use determinacy of the game to extract a `failed branch' and construct the countermodel from that.

\begin{proof}
[Proof of Proposition~\ref{prop:denier-subtree}]
If there is no proof of sequent $S = \relats R, \Gamma \seqar \Delta$, we know by determinacy of the game (Proposition \ref{prop:determinacy}) that there is a winning strategy for Denier ($\denier$) for $\relats R, \Gamma \seqar \Delta$. So, when navigating the proof search space from bottom-up, whenever we must make a choice at a branching rule in an invertible phase, we can simply follow the choice given by the winning strategy for $\denier$. Starting from the bottom sequent, we perform this strategy and for each non-invertible branching that we encounter, we continue the same process for each branch. As we follow the winning strategy of $\denier$, no infinite branch can be progressing.
\end{proof}

\begin{proof}[Proof of \cref{thm:countermodel-construction}]
We rely on the properties from Lemma~\ref{lem:invertiblefinite} and Proposition~\ref{prop:denier-subtree} for $S = (\relats R, \Gamma \seqar \Delta)$. By Lemma~\ref{lem:invertiblefinite}, the subtree $T$ found in Proposition~\ref{prop:denier-subtree} consists of countably many finite segments of the form $S_1 \dots S_n$ where $S_1$ is the start sequent or premiss of a non-invertible rule, and $S_n$ is saturated. We assume that the segments are indexed from index set $I$. 

To define the countermodel, we use the fact that for each segment $i$ of the form $S^i_1 \dots S^i_{n_i}$, we know that all its sequents are contained in $S^i_n = \relats R_i, \Gamma_i \seqar \Delta_i$. So for each such segment~$i$, consider $S^i_{n_i}$. We define the structure $\Kmod = (W,\leq,\{D_w\}_{w \in W}, \{\Prop_w\}_{w \in W}, \{ R_w \}_{w \in W})$ as follows:
\begin{align*}
    W & \ =  \ \{ w_i \mid i \in I \};\\
    w_i \leq w_j & \text{ \ if{f} } \ i = j \text{ or }\relats R_i,\Gamma_i \seqar \Delta_i \text{ becomes below } \relats R_j;\Gamma_j \seqar \Delta_j \text{ in the proof search tree};\\
    D_{w_i} & \ = \ \Lab(\relats R_i);\\
    xR_{w_i}y & \text{ \ if{f} } \ xRy \in \relats R_i;\\
    \Prop_{w_i}(p) & \ = \ \{ x \in \Lab(\relats R_i) \mid x : p \in \Gamma_i \}.
\end{align*}
It remains to show that $\Kmod$ is in $\predIGL$ and that it satisfies the conditions from the theorem. 

We first prove that $\Kmod$ is an intuitionistic structure. By construction of the rules in the countermodel construction, for each segments $i$ that is below a segment $j$ we have $\relats R_i \subseteq \relats R_j$ and $\Gamma_i \subseteq \Gamma_j$. Therefore, for each $w_i \leq w_j$, it holds that $D_{w_i} \subseteq D_{w_j}$, $\Prop_{w_i}(p) \subseteq \Prop_{w_j}(p)$ for each $p \in \Prop$, and $R_{w_i} \subseteq R_{w_j}$. Also note that $\leq$ is a partial order.

Each $R_{w_i}$ is transitive because sequents $\relats R_i, \Gamma_i \seqar \Delta_i$ are saturated by construction. In addition, relation $(\leq_{D_W} ; R_{D_W})$ as defined in Section \ref{sec:predmodels} is terminating as proved as follows. Suppose we have an infinite branch in the model:
\[
(w^1,y^1) \leq_{D_W} (w^2,y^1) R_{D_W}  (w^2,y^2) \leq_{D_W} (w^3,y^2) R_{D_W}  (w^3,y^3) \dots
\]
This means that there is an infinite branch in the strategy of $\denier$, say $(\relats R^k, \Gamma^k \seqar x^k : A^k)_k$ such that $y^k R y^{k+1} \in \relats R^k$ for infinitely many $k$. But this means that this branch has a progressing trace, which cannot be the case in the game strategy of $\denier$, a contradiction. Therefore, $(\leq_{D_W} ; R_{D_W})$ is terminating, and therefore $\Kmod \in \predIGL$.

Now we show that for all $w_i$-environments $\rho_{w_i}$ with $\rho_{w_i}(x)=x$ for $x \in \Lab(\relats R_i)$, we have for all labelled formulas $x : A$ and all $i \in I$:
\begin{align*}
    &\text{if } x : A \in \Gamma_i, \text{ then } \Kmod, w_i \models^{\rho_{w_i}} \STsub{x}{A}, \text{ and, }\\
    &\text{if } x : A \in \Delta_i, \text{ then } \Kmod, w_i \notmodels^{\rho_{w_i}} \STsub{x}{A}.
\end{align*}
We do so by induction on $A$. We show some cases and leave the remaining cases to the reader.
\begin{itemize}
    \item 
    For $A = p$, if $x : p \in \Gamma_i$, then $\Kmod, w_i \models^{\rho_{w_i}} x : p$ by definition of the structure and definition of $\rho$. 
    
    If $x : p \in \Delta_i$ we observe that $x : p \notin \Gamma_i$, otherwise sequent $\relats R_i, \Gamma_i \seqar \Delta_i$ would be provable which cannot be the case for a sequent in $\denier$'s strategy. Hence, $\Kmod, w_i \notmodels^{\rho_{w_i}} x : p$.
    
    \item 
    For $A = A_1 \wedge A_2$, 
    if $x : A_1 \wedge A_2 \in \Gamma_i$, then by saturation, $x : A_1 \in \Gamma_i$ and $x : A_2 \in \Gamma_i$. By induction hypothesis we have $\Kmod, w_i \models^{\rho_{w_i}} \STsub{x}{A_1}$ and $\Kmod, w_i \models^{\rho_{w_i}} \STsub{x}{A_2}$. Hence $\Kmod, w_i \models^{\rho_{w_i}} \STsub{x}{A_1 \wedge A_2}$. Similar reasoning applies when $x : A_1 \wedge A_2 \in \Delta_i$.
    
    \item 
    For $A = A_1 \impl A_2$, 
    if $x : A_1 \impl A_2 \in \Gamma_i$, let $w_j \geq w_i$ such that $\Kmod, w_j \models^{\rho_{w_i}} \STsub{x}{A_1}$. By the proof search, $\Gamma_i \subseteq \Gamma_j$, therefore $x : A_1 \impl A_2 \in \Gamma_j$. By saturation we have $x:A_1 \in \Delta_j$ or $x:A_2 \in \Gamma_j$. The first cannot occur as by induction hypothesis we would have $\Kmod, w_j \notmodels^{\rho_{w_j}} \STsub{x}{A_1}$ which is in contradiction to the assumption. (Note that $x : A_1$ has no other free variable than $x$ and indeed $\rho_{w_i}(x)=x=\rho_{w_j}(x)$.) By applying the induction hypothesis to the second we obtain $\Kmod, w_j \models^{\rho_{w_i}} \STsub{x}{A_2}$ as desired.

    If $x : A_1 \impl A_2 \in \Delta_i$, then for an immediate segment $j$ above segment $i$ in the proof search tree we have $x : A_1 \in \Gamma_j$ and $x : A_2 \in \Delta_j$. By induction hypothesis we have $\Kmod, w_j \models^{\rho_{w_j}} \STsub{x}{A_1}$ and $\Kmod, w_j \notmodels^{\rho_{w_j}} \STsub{x}{A_2}$ (and by similar reasoning from above $\Kmod, w_j \models^{\rho_{w_i}} \STsub{x}{A_1}$ and $\Kmod, w_j \notmodels^{\rho_{w_i}} \STsub{x}{A_2}$). So, since $w_i \leq w_j$, we have $\Kmod, w_i \notmodels^{\rho_{w_i}} \STsub{x}{A_1 \impl A_2}$.

    \item
    For $A = \Box A'$, 
    if $x : \Box A' \in \Gamma_i$, let $w_j \geq w_i$ and let $y \in D_{w_j} = \Lab(\relats R_j)$. Suppose 
    $R_{w_j}(\rho_{w_i}(x),y)$. Note that $\rho_{w_i}(x)=x$, and thus $xRy \in \relats R_j$.
    Note that $x : \Box A' \in \Gamma_j$. By saturation, $y : A' \in \Gamma_j$. By induction hypothesis we conclude $\Kmod, w_j \models^{\rho_{w_j}} \STsub{y}{A'}$. Since $\rho_{w_j}(y)=y$, we conclude $\Kmod, w_j \models^{\rho_{w_i}[y:=y]} \STsub{y}{A'}$. 

    If $x : \Box A' \in \Delta_i$, then for an immediate segment $j$ above $i$ in the proof search tree we have $xRy \in \relats R_j$ and $y : A' \in \Delta_j$. So by induction hypothesis $\Kmod, w_j \notmodels^{\rho_{w_j}} \STsub{y}{A'}$. Since $w_i \leq w_j$ and $x R_{w_j} y$, and since $\rho_{w_j}(y)=y$ we have $\Kmod, w_i \notmodels^{\rho_{w_i}[y:=y]} \STsub{y}{\Box A'}$. 
\end{itemize}
\end{proof}

\section{Appendix for \cref{sec:cut-elim}}
\label{app:cut-elim}

\begin{proof}
    [Proof of \cref{prop:cut-free-multi-to-dis-lab-cut-sys}]
    Replace the RHS of any $\mlIKfour$ $\infty$-proof by the disjunction of its labelled formulas (without repetitions), and locally simulate any right step by deriving the corresponding implication and cutting against it.
    Since cuts are inserted locally on only logical formulas, not interfering with relational atoms, the resulting preproof is indeed progressing.
\end{proof}

For \cref{lem:rel-cxt-from-bars-preserved,lem:push-cuts-above-bar}, we describe $\mathsf r(B)$ along with the cut-reduction cases in the next subsection, whence the arguments follow by inspection of the cases.

\begin{proof}
    [Proof of \cref{lem:invertibility-of-orleft}]
    Replace $\phi_0 \lor \phi_1$ and all its direct ancestors by $\phi_i$. The only critical cases are when it is principal, in which case simply take the $i$-premiss (and choose the corresponding branch).
    Formally, the argument is structured as a coinduction on $P$.
\end{proof}

\begin{proof}
    [Proof of \cref{prop:cut-degree-reduction}]
    Starting with a $\infty$-proof $P$, write $B_1, B_2,\dots$ for its bars of heights $1,2,\dots$ respectively.
    Define a sequence of cut-reductions  $P_0 \leadsto^*_{\vec \cutred_1} P_1 \leadsto^*_{\vec\cutred_2}  \cdots $ as follows:
    \begin{itemize}
        \item Set $P_0 := P$.
        \item $\vec \cutred_{n+1}$ is obtained by applying \cref{lem:rel-cxt-from-bars-preserved} to $P_n$ and $\vec \cutred_n \cdots \vec \cutred_1 (B_{n+1})$.
    \end{itemize}
    So, e.g., $\vec \cutred_1$ is obtained by applying \cref{lem:rel-cxt-from-bars-preserved} to $P_0$ and $B_1$; $\vec \cutred_2$ is obtained by applying \cref{lem:rel-cxt-from-bars-preserved} to $P_1$ and $\vec \cutred_1 (B_2)$; and so on.
    Since cut-reductions leave nodes beneath untouched, this process must approach a limit preproof, say $P_{\omega}$.

    We must now show that $P_\omega$ is progressing.
    Consider an infinite branch $(S_i)_{i<\omega}$ of $P_\omega$ and write $\relats R := \bigcup\limits_{i<\omega} \relcxt {S_i}$.
    Notice that, by \cref{lem:rel-cxt-from-bars-preserved}, $\relats R$ must contain the relational contexts of arbitrarily high sequents from $P$.
    Specifically, the maximal prefix of $(S_i)_{i<\omega}$ consistent with $P_{n+1}$ ends at $\vec \cutred_n \cdots \vec \cutred_1 (B_{n+1})$ and so at a sequent whose relational context contains that of some sequent in $B_{n+1}$, i.e.\ at height $n$, by \cref{lem:rel-cxt-from-bars-preserved}.

    Finally, since relational contexts are growing, bottom-up, notice that the set of sequents of $P$ whose relational contexts are contained in $\relats R$ forms a subtree $T$ of $P$.
    Since, as just discussed, $T$ contains arbitrarily high nodes of $P$, it is infinite and so has an infinite branch by K\"onig's Lemma.
    Thus $T$ has a progressing trace and so by definition $\relats R$ has an infinite path.
    Again by definition, this means $(S_i)_{i<\omega}$ has a progressing trace.
\end{proof}

\subsection{Cut-reduction cases and $\mathsf r(B)$}

\subsubsection{Key case}

This is the key cut-reduction on $\phi$-formulas again:
\[
\small
\vlderivationnc{
\vliin{\cut}{}{\relats R, \Gamma, \Gamma' \seqar \phi}{
    \vlin{\rightrule \lor}{}{\relats R, \Gamma \seqar \chi_0 \lor \chi_1}{
    \vltr{P}{\relats R, \Gamma, \seqar \chi_i }{\vlhy{\ }}{\vlhy{\ }}{\vlhy{\ }}
    }
}{
    \vltr{Q}{\relats R, \Gamma', \chi_0 \lor \chi_1 \seqar \phi}{\vlhy{\ }}{\vlhy{\ }}{\vlhy{\ }}
}
}
\qquad \leadsto \qquad 
\vlderivationnc{
\vliin{\cut}{}{\relats R, \Gamma, \Gamma' \seqar \phi}{
    \vltr{P}{\relats R, \Gamma \seqar \chi_i }{\vlhy{\ }}{\vlhy{\ }}{\vlhy{\ }}
}{
    \vltr{Q_i}{\relats R, \Gamma' , \chi_i \seqar \phi}{\vlhy{\ }}{\vlhy{\ }}{\vlhy{\ }}
}
}
\]
where $Q_i$ is obtained by \cref{lem:invertibility-of-orleft} above.

Here, if $B$ intersects $P$ and $Q$ before reduction then $\mathsf r(B)$ consists of the corresponding points after reduction.
If $B$ intersects the premiss of the cut or $\rightrule \lor$ before reduction then $\cutred (B)$ consists of the conclusion of $P$ after reduction, as well as the corresponding part of $Q$ before the reduction. 
Note in all these cases that the degree below the bar has decreased.
If $B$ is below (or including) the conclusion of the cut before reduction, then $\cutred(B)$ is the corresponding bar after reduction, i.e.\ the cut-reduction leaves the bar unchanged. 

The definition of $\cutred(B)$ is similar for all the commutative cases below, always designed to decrease the distance of a topmost $d$-cut from the bar.

\subsubsection{Commutative cases}

This is the commutation over a $\leftrule\limp $ step,
\[
\footnotesize
\vlderivationnc{
\vliin{\cut}{}{\relats R, \Gamma, \Gamma', \Gamma'', x:A\limp B \seqar \phi }{
    \vliin{\leftrule \limp}{}{\relats R, \Gamma, \Gamma', x:A\limp B \seqar \chi}{
        \vltr{P}{\relats R, \Gamma \seqar x:A}{\vlhy{\ }}{\vlhy{\ }}{\vlhy{\ }}
    }{
        \vltr{Q}{\relats R, \Gamma', x:B \seqar \chi}{\vlhy{\ }}{\vlhy{\ }}{\vlhy{\ }}
    }
}{
    \vltr{R}{\relats R, \Gamma'', \chi \seqar \phi}{\vlhy{\ }}{\vlhy{\ }}{\vlhy{\ }}
}
}
\quad \leadsto \quad
\vlderivationnc{
\vliin{\leftrule \limp}{}{\relats R, \Gamma, \Gamma', \Gamma'', x:A\limp B \seqar \phi}{
    \vltr{P}{\relats R, \Gamma \seqar x:A}{\vlhy{\ }}{\vlhy{\ }}{\vlhy{\ }}
}{
    \vliin{\cut}{}{\relats R, \Gamma', \Gamma'', x:B \seqar \phi}{
        \vltr{Q}{\relats R, \Gamma', x:B \seqar \chi}{\vlhy{\ }}{\vlhy{\ }}{\vlhy{\ }}
    }{
        \vltr{R}{\relats R, \Gamma'', \chi \seqar \phi}{\vlhy{\ }}{\vlhy{\ }}{\vlhy{\ }}
    }
}
}
\]

This is the commutation over a  $\leftrule \lor$ step,
\[
\footnotesize
\vlderivationnc{
\vliin{\cut}{}{\relats R, \Gamma, \Gamma' , \psi_0 \lor \psi_1 \seqar \phi}{
    \vliin{\leftrule \lor}{}{\relats R, \Gamma, \psi_0 \lor \psi_1 \seqar \chi}{
        \vltr{P_0}{\relats R, \Gamma, \psi_0 \seqar \chi}{\vlhy{\ }}{\vlhy{\ }}{\vlhy{\ }}
    }{
        \vltr{P_1}{\relats R, \Gamma, \psi_1 \seqar \chi}{\vlhy{\ }}{\vlhy{\ }}{\vlhy{\ }}
    }
}{
    \vltr{Q}{\relats R, \Gamma' , \chi \seqar \phi}{\vlhy{\ }}{\vlhy{\ }}{\vlhy{\ }}
}
}
\quad \leadsto \quad
\vlderivationnc{
\vliin{\leftrule \lor}{}{\relats R, \Gamma, \Gamma' , \psi_0 \lor \psi_1 \seqar \phi }{
    \vliin{\cut}{}{\relats R, \Gamma, \Gamma' , \psi_0 \seqar \phi}{
        \vltr{P_0}{\relats R, \Gamma, \psi_0 \seqar \chi}{\vlhy{\ }}{\vlhy{\ }}{\vlhy{\ }}
    }{
        \vltr{Q}{\relats R, \Gamma' , \chi \seqar \phi}{\vlhy{\ }}{\vlhy{\ }}{\vlhy{\ }}
    }
}{
    \vliin{\cut}{}{\relats R, \Gamma, \Gamma', \psi_1 \seqar \phi}{
        \vltr{P_1}{\relats R, \Gamma, \psi_1 \seqar \chi}{\vlhy{\ }}{\vlhy{\ }}{\vlhy{\ }}
    }{
        \vltr{Q}{\relats R, \Gamma' , \chi \seqar \phi}{\vlhy{\ }}{\vlhy{\ }}{\vlhy{\ }}
    }
}
}
\]

This is the commutation over a $\leftrule \Diamond$ step:
\[
\footnotesize
\vlderivationnc{
\vliin{\cut}{}{\relats R, \Gamma, \Gamma', x:\Diamond A \seqar \phi}{
    \vlin{\leftrule \Diamond}{}{\relats R, \Gamma, x: \Diamond A \seqar \chi}{
    \vltr{P}{\relats R, xRy, \Gamma , y:A \seqar \chi}{\vlhy{\ }}{\vlhy{\ }}{\vlhy{\ }}
    }
}{
    \vltr Q {\relats R, \Gamma' , \chi \seqar \phi}{\vlhy{\ }}{\vlhy{\ }}{\vlhy{\ }}
}
}
\quad \leadsto \quad
\vlderivationnc{
\vlin{\leftrule \Diamond}{}{\relats R, \Gamma, \Gamma' , x:\Diamond A \seqar \phi}{
\vliin{\cut}{}{\relats R, xRy, \Gamma, \Gamma', y:A \seqar \phi}{
    \vltr{P}{\relats R, xRy, \Gamma , y:A \seqar \chi}{\vlhy{\ }}{\vlhy{\ }}{\vlhy{\ }}
}{
    \vltr {xRy,Q} {xRy,\relats R, \Gamma' , \chi \seqar \phi}{\vlhy{\ }}{\vlhy{\ }}{\vlhy{\ }}
}
}
}
\]
where $xRy,Q$ is obtained from $Q$ by prepending $xRy$ to the LHS of each sequent.

This is the commutation over a $\leftrule \Box$ step:
\[
\small
\vlderivationnc{
\vliin{\cut}{}{\relats R, xRy, \Gamma, \Gamma', x:\Box A \seqar \phi}{
    \vlin{\leftrule \Box}{}{\relats R, xRy, \Gamma , x:\Box A \seqar \chi}{
    \vltr{P}{\relats R, xRy,\Gamma, y: A \seqar \chi }{\vlhy{\ }}{\vlhy{\ }}{\vlhy{\ }}
    }
}{
    \vltr{Q}{\relats R, xRy, \Gamma' , \chi\seqar \phi }{\vlhy{\ }}{\vlhy{\ }}{\vlhy{\ }}
}
}
\quad \leadsto\quad
\vlderivationnc{
\vlin{\leftrule \Box}{}{\relats R, xRy, \Gamma, \Gamma' , x:\Box A \seqar \phi}{
\vliin{\cut}{}{\relats R, xRy, \Gamma, \Gamma', y:A \seqar \phi}{
    \vltr{P}{\relats R, xRy,\Gamma, y: A \seqar \chi }{\vlhy{\ }}{\vlhy{\ }}{\vlhy{\ }}
}{
    \vltr{Q}{\relats R, xRy, \Gamma' , \chi\seqar \phi }{\vlhy{\ }}{\vlhy{\ }}{\vlhy{\ }}
}
}
}
\]

This is the commutation over a $\leftrule\contr$ step:
\[
\vlderivationnc{
\vliin{\cut}{}{\relats R, \Gamma, \psi ,\Gamma' \seqar \phi}{
    \vlin{\leftrule \contr}{}{\relats R, \Gamma , \psi \seqar \chi}{
    \vltr{P}{\relats R, \Gamma, \psi, \psi \seqar \chi}{\vlhy{\ }}{\vlhy{\ }}{\vlhy{\ }}
    }
}{
    \vltr{Q}{\relats R, \Gamma', \chi \seqar \phi}{\vlhy{\ }}{\vlhy{\ }}{\vlhy{\ }}
}
}
\quad\leadsto\quad
\vlderivationnc{
\vlin{\leftrule \contr}{}{\relats R, \Gamma, \psi, \Gamma' \seqar \phi}{
\vliin{\cut}{}{\relats R, \Gamma, \psi, \psi, \Gamma' \seqar \phi}{
    \vltr{P}{\relats R, \Gamma, \psi, \psi \seqar \chi}{\vlhy{\ }}{\vlhy{\ }}{\vlhy{\ }}
}{
    \vltr{Q}{\relats R, \Gamma', \chi \seqar \phi}{\vlhy{\ }}{\vlhy{\ }}{\vlhy{\ }}
}
}
}
\]
The commutations over $\leftrule \wk$ and $\leftrule \land$ are similar.

\end{document}